\documentclass[sn-mathphys-num]{sn-jnl}

\usepackage{graphicx}%
\usepackage{multirow}%
\usepackage{amsmath,amssymb,amsfonts}%
\usepackage{amsthm}%
\usepackage{mathrsfs}%
\usepackage{mathtools}
\usepackage[title]{appendix}%
\usepackage{xcolor}%
\let\mathscr\mathrsfs
\usepackage{textcomp}%
\usepackage{manyfoot}%
\usepackage{algorithm}%
\usepackage{algorithmicx}%
\usepackage{algpseudocode}%
\usepackage{listings}%
\usepackage{booktabs}
\usepackage{rotating}
\usepackage{tabularx} 
\theoremstyle{thmstyleone}%
\newtheorem{theorem}{Theorem}
\newtheorem{proposition}[theorem]{Proposition}%
\newtheorem{lemma}[theorem]{Lemma}%
\newtheorem{corollary}[theorem]{Corollary}%

\theoremstyle{thmstyletwo}%
\newtheorem{example}{Example}%
\newtheorem{remark}{Remark}%

\theoremstyle{thmstylethree}%
\newtheorem{definition}{Definition}%

\begin{document}

\title[Article Title]{Asymmetric Entanglement-Assisted Quantum Locally Recoverable Codes}


\author*[1]{\fnm{Vijay} \sur{Kumar}}\email{vijayk@iiitnr.edu.in}

\author[2]{\fnm{Ramakrishna} \sur{Bandi}}\email{ramakrishna@iiitnr.edu.in}
\equalcont{These authors contributed equally to this work.}

\affil[1,2]{\orgdiv{Department of Science and Mathematics}, \orgname{International Institute of Information Technology },\orgaddress{ \city{Naya Raipur}, \postcode{493661}, \state{Chhattisgarh}, \country{India}}}










\abstract{Locally recoverable codes (LRCs) permit reconstruction of a lost coordinate from a small subset of the remaining coordinates, a property that is a fundamental coding technique for distributed storage and, more recently, for fault-tolerant quantum memories and distributed storage systems. Existing constructions of quantum locally recoverable codes (qLRCs) and their entanglement-assisted generalizations (EA-qLRCs) impose a single locality parameter $r$ governing recovery from both bit-flip ($X$-type) and phase-flip ($Z$-type) errors, even though physical noise is rarely symmetric between the two. We introduce entanglement-assisted quantum locally recoverable codes (EA-qLRC) with \emph{asymmetric locality} $\{r_X,r_Z\}$, in which $X$-type and $Z$-type errors are corrected from recovery sets of size at most $r_X+1$ and $r_Z+1$, respectively. We extend a stabilizer sufficient condition for this property and, in particular, the CSS framework, where asymmetric locality corresponds to separate locality conditions on the two classical codes $\mathcal C_1$ and $\mathcal C_2$ underlying the construction. From this characterization we derive Singleton-, Griesmer-, Plotkin-,
sphere-packing-, and Cadambe\textendash Mazumdar-like bounds for CSS-type asymmetric EA-qLRCs (AEA-qLRCs) built from Proposition~\ref{CSS-EA-qLRC-Asym}, and we characterize exactly when a pure CSS-type AEA-qLRC built from a pair
of classical LRCs attains the Singleton-like bound with equality.}

\keywords{Quantum error correction, entanglement-assisted quantum codes, locally recoverable codes, asymmetric quantum codes, CSS construction, Singleton bound, Tamo-Barg codes, cyclic codes.}

\maketitle

\section{Introduction}

The rapid growth of large-scale distributed storage systems has stimulated extensive research on coding techniques that enable reliable data storage while reducing the cost of node repair. Classical locally recoverable codes (LRCs), introduced by Gopalan \emph{et al.}~\cite{gopalan2012locality}, address this need by allowing each erased symbol to be recovered by accessing only a small number of other symbols. Owing to their practical importance, LRCs have become an active area of research, with extensive studies on their fundamental bounds, explicit constructions, and optimality conditions~\cite{cadambe2015bounds,tamo2016bounds,papailiopoulos2014locally,tamo2015cyclic,xing2022construction}.

With the continued development of quantum computing, quantum communication, and quantum storage technologies, protecting quantum information efficiently has become increasingly important. Quantum error-correcting codes (QECCs) provide the fundamental mechanism for protecting quantum states against decoherence and operational errors. Motivated by the success of classical LRCs, Golowich \emph{et al.}~\cite{Golowic} introduced quantum locally recoverable codes (qLRCs), which enable an erased qudit to be recovered by accessing only a limited number of other qudits. Since then, qLRCs have attracted considerable attention because of their potential applications in quantum distributed storage, where locality can reduce the overhead associated with quantum recovery operations. Building on this foundation, several constructions~\cite{luo2025bounds,li2025improved,sharma2025quantum,bu2025quantum,xie2025two} have developed qLRCs through CSS and Hermitian frameworks by exploiting classical locally recoverable codes. In parallel, Singleton-like, Plotkin-like, Griesmer-like, and sphere-packing-type bounds have been established for symmetric qLRCs and their generalizations. These bounds and constructions provide important benchmarks for evaluating the optimality of qLRCs and have substantially advanced the understanding of quantum locality.

This framework has subsequently been extended in two directions: entanglement assistance and error asymmetry. The first direction concerns the use of pre-shared entanglement to relax the structural constraints arising in quantum code constructions. Li \emph{et al.}~\cite{li2026entanglement} generalized qLRCs to entanglement-assisted quantum locally recoverable codes (EA-qLRCs). Independently, Kumar \emph{et al.}~\cite{kumar2026entanglement} developed a systematic framework for CSS-like entanglement-assisted quantum locally recoverable codes (CSS-like EA-qLRCs). They established five fundamental converse bounds, namely Singleton-, Griesmer-, Plotkin-, sphere-packing-, and Cadambe\textendash Mazumdar-like bounds, and analyzed their relative asymptotic and regime-specific tightness. They also characterized necessary and sufficient conditions for pure codes to attain the Singleton-like bound, presented explicit optimal cyclic constructions, and identified a structural obstruction showing that Tamo\textendash Barg-based constructions attain this bound only in the vacuous $k\le r$ regime. Furthermore, they constructed maximally entangled codes from complementary duals, established two Gilbert\textendash Varshamov-like achievability bounds through parity-check augmentation and concatenation, and unified the converse and achievability bounds in the maximally entangled regime to characterize the trade-offs among rate, distance, and locality.

The second direction concerns asymmetric error correction. Since quantum noise can be asymmetric, with phase-flip ($Z$-type) errors occurring more frequently than bit-flip ($X$-type) errors~\cite{steane1996multiple}, asymmetric quantum codes (AQCs) have been developed to provide different error-correction capabilities for the two error types~\cite{ioffe2007asymmetric,sarvepalli2009asymmetric,wang2010asymmetric,ezerman2013css}. Building on this motivation, a companion work~\cite{li2026quantum} introduced quantum locally recoverable codes with asymmetric locality ${r_X,r_Z}$, where $r_X$ and $r_Z$ denote the numbers of qudits required to recover from $X$-type and $Z$-type single-qudit errors, respectively. The work adapted the completely depolarizing channel used in earlier qLRC constructions into separate $X$-type and $Z$-type Pauli twirling channels and established a CSS criterion under which $C=\mathrm{CSS}(C_X,C_Z)$ is an ${r_X,r_Z}$-qLRC whenever $C_X$ has locality $r_Z$ and $C_Z$ has locality $r_X$.

These two directions naturally raise the question of whether entanglement assistance can be incorporated into the asymmetric local-recovery framework. Entanglement-assisted quantum error-correcting codes (EAQECCs), introduced by Brun \emph{et al.}~\cite{brun2006correcting}, use pre-shared maximally entangled pairs between the encoder and decoder to remove the dual-containing constraint inherent in the standard CSS construction. Consequently, an EAQECC can be constructed from any pair of classical linear codes at the cost of consuming $c$ ebits. Galindo \emph{et al.}~\cite{galindo2020asymmetric} introduced asymmetric EAQECCs constructible from arbitrary pairs of classical linear codes over arbitrary fields, together with a Gilbert\textendash Varshamov-type existence bound and an explicit BCH-code construction. Liu \emph{et al.}~\cite{liu2021asymmetric} subsequently improved these bounds and provided further constructions for asymmetric EAQECCs.

However, it remains unclear whether this compatibility persists once local recoverability is imposed as an additional requirement. In particular, the question of whether entanglement assistance can overcome the dual-containment obstruction inherent in CSS-type constructions, while simultaneously preserving distinct locality requirements for $X$-type and $Z$-type errors, has not been addressed in the literature. The central problem motivating this work is therefore the following: \emph{does there exist a unified framework in which entanglement assistance, asymmetric error correction, and local recoverability can be attained simultaneously?}

\section*{Our Contribution}
In this paper, we introduce  \emph{asymmetric entanglement-assisted quantum locally recoverable codes} (AEA-qLRCs), a framework that simultaneously incorporates asymmetric error protection and local recoverability into the entanglement-assisted stabilizer formalism. Our contributions can be summarized as follows.

\begin{enumerate}
    \item \textbf{Definition and stabilizer characterization.} We define an AEA-qLRC with locality $\{r_X,r_Z\}$, requiring $X$-type and $Z$-type errors at each coordinate to be independently recoverable from sets of size at most $r_X+1$ and $r_Z+1$, respectively, using shared entanglement. We prove that a pair of stabilizer generators reducing, on the encoder side, to a pure $X$-type and a pure $Z$-type operator at coordinate $i$, with supports of size at most $r_X+1$ and $r_Z+1$, suffices for this property (Theorem~\ref{Stb_EAqLRC_Asym}).

    \item \textbf{CSS-type construction.} From this characterization we obtain a construction (Proposition~\ref{CSS-EA-qLRC-Asym}) yielding an AEA-qLRC with locality $\{r_X,r_Z\}$ from any two classical codes $\mathcal{C}_1,\mathcal{C}_2$ whose duals admit, at every coordinate, parity checks of weight at most $r_Z+1$ and $r_X+1$, respectively. The two localities are thus controlled independently by the two codes.

    \item \textbf{Bounds.} We prove a general trade-off theorem (Theorem~\ref{Bounds_theorem_Asym}) bounding $d_X+d_Z$, $\kappa$, and $n$ via the optimal-parameter functions $n_{\mathrm{opt}}^{(q)}$, $k_{\mathrm{opt}}^{(q)}$, $d_{\mathrm{opt}}^{(q)}$ of classical LRCs. Applying known classical bounds yields the Singleton-, Cadambe-Mazumdar-, Plotkin-, sphere-packing-, and Griesmer-like bounds for CSS-type AEA-qLRCs (Corollaries~\ref{cor:singleton_Asym_direct} and \ref{cor:singleton_Asym}).
    \item \textbf{Optimality.} We characterize Singleton-optimal pure AEA-qLRCs: equality holds if and only if $\mathcal{C}_1,\mathcal{C}_2$ are optimal LRCs with $d_Z=d_X$ and matching locality-ceiling conditions (Theorem~\ref{pure_EA-qLRC-Asym}). We restate this criterion in terms of $s_1=\dim(\mathcal{C}_1\cap\mathcal{C}_2^\perp)$ and $s_2=\dim(\mathcal{C}_2\cap\mathcal{C}_1^\perp)$ (Theorem~\ref{cor:two-code-optimality-asym}), and show that optimal \emph{maximally} entangled codes correspond exactly to Linear Complementary Pairs of optimal LRCs (Theorem~\ref{EA-qLRC-LCP-Asym}).

    \item \textbf{Constructions.} We instantiate the framework using two classes of classical LRCs. First, for Tamo--Barg codes, we explicitly evaluate the relevant dual-intersection sizes in Proposition~\ref{prop:TBintersection-asym}, leading to pure AEA-qLRC constructions in both the disjoint- and overlapping-support regimes; see Theorems~\ref{thm:TB-EAqLRC-Asym-1} and \ref{thm:TB-EAqLRC-Asym-2}. Second, for cyclic LRCs, we establish a support-alignment lemma (Lemma~\ref{lem:support-alignment-asym}) and use it to derive a general construction in Theorem~\ref{Cyclic-EA-qLRC-Asym}. Furthermore, by considering linear complementary pairs of optimal cyclic LRCs, we obtain an optimal maximally entangled family of AEA-qLRCs in Theorem~\ref{Cyclic-MEAQECC-Asym}.

\end{enumerate}
\subsection*{Relation to Concurrent EA-qLRC Work}

This work builds directly on two concurrent lines of research: entanglement-assisted
quantum locally recoverable codes (EA-qLRCs) with a single, symmetric locality
parameter $r$, developed independently by Li \emph{et al.}~\cite{li2026entanglement}
and Kumar and Bandi~\cite{kumar2026entanglement}; and quantum locally recoverable
codes with asymmetric locality $\{r_X,r_Z\}$ but without entanglement assistance,
developed by Li, Jin, and Xing~\cite{li2026quantum}. The present paper unifies these
two directions by introducing independent localities $\{r_X,r_Z\}$ into the
entanglement-assisted stabilizer formalism. Table~\ref{tab:novelty} summarizes,
result by result, how the contributions of this paper relate to
\cite{li2026entanglement,kumar2026entanglement,li2026quantum}. In particular, setting
$r_X=r_Z=r$ throughout recovers exactly the stabilizer criterion, CSS construction,
and bounds of \cite{li2026entanglement,kumar2026entanglement}, so the results below
should be understood as a strict asymmetric generalization of that framework rather
than an independent development.
\begin{table}[!ht]
\centering
\caption{Relation of this paper's results to concurrent EA-qLRC and asymmetric qLRC work.}\label{tab:novelty}
\begin{tabularx}{\textwidth}{XXXX}
\toprule
\textbf{Result} & \textbf{\cite{li2026entanglement,kumar2026entanglement}} & \textbf{\cite{li2026quantum}} & \textbf{This paper} \\
\midrule
Entanglement assistance            & Yes            & No              & Yes \\
Asymmetric locality $\{r_X,r_Z\}$  & No ($r$ only)  & Yes             & Yes \\
Stabilizer sufficient condition    & Single $r$     & Single $\{r_X,r_Z\}$, no EA & Independent $\{r_X,r_Z\}$ with EA (Thm.~\ref{Stb_EAqLRC_Asym}) \\
CSS-type construction              & Single $r$     & Single $\{r_X,r_Z\}$, no EA & Independent support bounds $r_Z\!+\!1$, $r_X\!+\!1$ (Prop.~\ref{CSS-EA-qLRC-Asym}) \\
Singleton-like bound               & Symmetric      & Asymmetric, no EA & Asymmetric with EA (Cor.~\ref{cor:singleton_Asym_direct}) \\
Griesmer/Plotkin/sphere-packing/CM & Symmetric      & --              & Asymmetric with EA (Cor.~\ref{cor:singleton_Asym}) \\
Pure-code optimality               & Single ceiling condition & --   & Two independent ceiling conditions (Thm.~\ref{pure_EA-qLRC-Asym}) \\
LCP / maximally entangled          & Symmetric LCP  & --              & Asymmetric LCP, requires $d_X=d_Z$ (Thm.~\ref{EA-qLRC-LCP-Asym}) \\
Tamo--Barg construction            & Symmetric $r$  & --              & Independently tunable $r_X,r_Z$ (Thm.~\ref{thm:TB-EAqLRC-Asym-1}, \ref{thm:TB-EAqLRC-Asym-2}) \\
Cyclic construction                & Symmetric $r$  & --              & Independently tunable $r_X,r_Z$ (Thm.~\ref{Cyclic-EA-qLRC-Asym}, \ref{Cyclic-MEAQECC-Asym}) \\
\bottomrule
\end{tabularx}
\end{table}

\section*{Organization}

The remainder of the paper is organized to mirror our contributions above. Section~\ref{sec:preliminary} recalls the necessary background\textemdash on classical linear and cyclic codes, classical locally recoverable codes and their bounds, the Tamo--Barg and cyclic LRC constructions, the entanglement-assisted stabilizer formalism, entanglement-assisted quantum locally recoverable codes (EA-qLRCs), and asymmetric entanglement-assisted quantum error-correcting codes (AEAQECCs). Section~\ref{sec:AEA-qLRCs} introduces AEA-qLRCs, establishes their stabilizer sufficient condition (Theorem~\ref{Stb_EAqLRC_Asym}), and gives the associated CSS-type construction (Proposition~\ref{CSS-EA-qLRC-Asym}). Section~\ref{sec:bounds} derives parameter bounds (Theorem~\ref{Bounds_theorem_Asym}, Corollaries~\ref{cor:singleton_Asym_direct}--\ref{cor:singleton_Asym}) and establishes optimality characterizations (Theorems~\ref{pure_EA-qLRC-Asym}, \ref{cor:two-code-optimality-asym}, \ref{EA-qLRC-LCP-Asym}). Section~\ref{sec:construction} presents explicit constructions from Tamo--Barg codes (Proposition~\ref{prop:TBintersection-asym}, Theorems~\ref{thm:TB-EAqLRC-Asym-1}, \ref{thm:TB-EAqLRC-Asym-2}) and from cyclic codes (Lemma~\ref{lem:support-alignment-asym}, Theorems~\ref{Cyclic-EA-qLRC-Asym}, \ref{Cyclic-MEAQECC-Asym}). Section~\ref{sec:conclusion} concludes with future research work.

\section{Preliminaries}\label{sec:preliminary}
Throughout, let $q$ be a prime power, and let $\mathbb F_q$ denote the finite field of order $q$. For $m,n\in\mathbb N$, define $[m]^\dagger:=\{0,1,\ldots,m\},\;[n]:=\{1,\ldots,n\}$.
For $\mathbf u=(u_1,\ldots,u_n)\in\mathbb F_q^n$, let $\operatorname{supp}(\mathbf u)=\{i\in[n]:u_i\neq 0\}$ denote its support and $\operatorname{wt}(\mathbf u)=|\operatorname{supp}(\mathbf u)|$ its Hamming weight. For a subset $S\subseteq\mathbb F_q^n$, define
$\operatorname{wt}(S)=\min_{\substack{\mathbf u\in S\\\mathbf u\neq\mathbf 0}}\operatorname{wt}(\mathbf u)$. In particular, for a linear code $\mathcal C$, $\operatorname{wt}(\mathcal C)$ denotes the minimum weight of a nonzero codeword. We use $\operatorname{tr}(\cdot)$ to denote the trace map from the relevant extension field to $\mathbb F_p$, where $p$ is the characteristic of $\mathbb F_q$, and $\omega=e^{2\pi i/p}.$

\subsection{Classical Linear Codes}
Let $\mathbb{F}_q$ denote the finite field of order $q$. A linear code of length $n$ and dimension $k$ over $\mathbb{F}_q$ is a $k$-dimensional subspace $\mathcal{C} \subseteq \mathbb{F}_q^n$, denoted as an $[n, k]_q$ code. The minimum distance of $\mathcal{C}$ is $d = \min_{c \in \mathcal{C} \setminus \{0\}} \operatorname{wt}(c)$. A code with these parameters is denoted as $[n, k, d]_q$.A generator matrix of an $[n,k]_q$ code $\mathcal{C}$ is a matrix $G \in \mathbb{F}_q^{k \times n}$ of rank $k$ whose row space equals $\mathcal{C}$. Equivalently:
$$\mathcal{C} = \{\, uG : u \in \mathbb{F}_q^k \,\}$$Equip $\mathbb{F}_q^n$ with the standard Euclidean bilinear form $\langle x, y \rangle_E = \sum_{i=1}^n x_i y_i$. The Euclidean dual of $\mathcal{C}$ is:$$\mathcal{C}^{\perp} = \{\, x \in \mathbb{F}_q^n : \langle x, c \rangle_E = 0 \ \ \forall\, c \in \mathcal{C} \,\}$$For a code $\mathcal{C}$ over $\mathbb{F}_{q^2}$, equip $\mathbb{F}_{q^2}^n$ with the conjugate-linear Hermitian form $\langle x, y \rangle_H = \sum_{i=1}^n x_i y_i^q$, where $z \mapsto z^q$ is the nontrivial field automorphism of $\mathbb{F}_{q^2}$ over $\mathbb{F}_q$. The Hermitian dual of $\mathcal{C}$ is:$$\mathcal{C}^{\perp_H} = \{\, x \in \mathbb{F}_{q^2}^n : \langle x, c \rangle_H = 0 \ \ \forall\, c \in \mathcal{C} \,\}$$If $\mathcal{C}$ is an $[n, k, d]_q$ code, its Euclidean dual $\mathcal{C}^{\perp}$ is an $[n, n-k, d^{\perp}]_q$ code. The dual minimum distance $d^{\perp} = \min_{x \in \mathcal{C}^{\perp} \setminus \{0\}} \operatorname{wt}(x)$ is independent of $d$.A matrix $H \in \mathbb{F}_q^{(n-k) \times n}$ is a parity-check matrix of $\mathcal{C}$ if $\operatorname{rank}(H) = n-k$ and:$$\mathcal{C} = \ker(H) = \{\, x \in \mathbb{F}_q^n : Hx^{\top} = \mathbf{0} \,\}$$Equivalently, the rows of $H$ form a basis of $\mathcal{C}^{\perp}$. For any generator matrix $G$ and parity-check matrix $H$ of $\mathcal{C}$:$$GH^{\top} = 0$$A generator matrix $G$ of $\mathcal{C}$ is systematic if:$$G = [\,I_k \mid A\,]$$where $I_k$ is the $k \times k$ identity matrix and $A \in \mathbb{F}_q^{k \times (n-k)}$. In this case,$$H = [\, -A^{\top} \mid I_{n-k} \,]$$is a parity-check matrix of $\mathcal{C}$. Encoding via a systematic $G$ maps a message $u \in \mathbb{F}_q^k$ to a codeword whose first $k$ coordinates equal $u$ and whose remaining $n-k$ coordinates are parity symbols. Let $\mathcal{C}$ be an $[n,k]_q$ code and $\mathcal{C}^{\perp}$ its Euclidean dual. Then $\mathcal{C}$ is called self-orthogonal if $\mathcal{C} \subseteq \mathcal{C}^{\perp}$, self-dual if $\mathcal{C} = \mathcal{C}^{\perp}$, dual-containing if $\mathcal{C}^{\perp} \subseteq \mathcal{C}$ and a linear complementary dual (LCD) code if $\mathcal{C} \cap \mathcal{C}^{\perp} = \{0\}$. Every $[n, k, d]_q$ code satisfies the Singleton bound $d \le n - k + 1$. A code attaining $d = n - k + 1$ is called a maximum distance separable (MDS) code.

\begin{definition}
Let $\gcd(n,q)=1$, let $m=\operatorname{ord}_n(q)$, and let $\alpha$ be a primitive $n$-th root of unity in $\mathbb F_{q^m}$. For $0\le i<n$, the \emph{minimal polynomial} of $\alpha^i$ over $\mathbb F_q$ is the unique monic irreducible polynomial $M_i(x)\in\mathbb F_q[x]$ having $\alpha^i$ as a root. Equivalently,
\[
M_i(x)=\prod_{j\in C_i}(x-\alpha^j),
\]
where $C_i=\{iq^t\bmod n:t\ge0\}$ is the $q$-cyclotomic coset containing $i$.
\end{definition}

A linear code $\mathcal C$ of length $n$ over $\mathbb F_q$ is called a \emph{cyclic code} if it is invariant under cyclic shifts; that is, whenever $(c_0,c_1,\ldots,c_{n-1})\in\mathcal C$
then $(c_{n-1},c_0,\ldots,c_{n-2})\in\mathcal C$.
Equivalently, identifying $\mathbb F_q^n$ with the quotient ring
\[
R_n=\mathbb F_q[x]/\langle x^n-1\rangle,
\]
every cyclic code corresponds to an ideal of $R_n$ and is uniquely generated by a monic polynomial $g(x)$ satisfying $g(x)\mid(x^n-1)$. The polynomial $g(x)$ is called the \emph{generator polynomial} of $\mathcal C$, and
\[
k=n-\deg(g(x)).
\]

If $\mathcal C$ is a cyclic code of length $n$ with generator polynomial $g(x)$, its \emph{defining set} is
\[
Z=\{i\in\{0,1,\ldots,n-1\}:g(\alpha^i)=0\}.
\]
The defining set $Z$ is a union of $q$-cyclotomic cosets. Moreover,
\[
g(x)=\operatorname{lcm}\{M_i(x):i\in Z\},
\]
that is, the generator polynomial is the least common multiple of the minimal polynomials corresponding to the exponents in $Z$.

\begin{lemma}[BCH Bound {\cite{bose1960class,hocquenghem1959codes}}]
Let $\mathcal C$ be a cyclic code of length $n$ over $\mathbb F_q$ with defining set $Z$. If $Z$ contains $\delta-1$ consecutive integers, then the minimum Hamming distance of $\mathcal C$ satisfies $d(\mathcal C)\ge\delta.$
\end{lemma}

\subsection{Classical Locally Recoverable Codes}

We next recall the notion of locality for classical linear codes, which forms the basis for the locally recoverable codes considered throughout this work.

\begin{definition}
Let $\mathcal C$ be an $[n,k,d]_q$ linear code. A coordinate $i\in[n]$ is said to have \emph{locality} $r$ if there exists a recovery set $\Gamma_i\subseteq[n]$ containing $i$, with
$|\Gamma_i|\le r+1$, such that the symbol at coordinate $i$ can be recovered from the symbols indexed by $\Gamma_i\setminus\{i\}$ for every codeword in $\mathcal C$. Equivalently, there exists a parity-check vector $\mathbf h\in\mathcal C^\perp$  such that $i\in\operatorname{supp}(\mathbf h)$, $|\operatorname{supp}(\mathbf h)|\le r+1$. In this case, one may take
$\Gamma_i=\operatorname{supp}(\mathbf h)$. The code $\mathcal C$ is called a \emph{locally recoverable code (LRC)} with locality $r$ if every coordinate of $\mathcal C$ has locality $r$. We also refer to such a code as an $(n,k,d,r)_q$-cLRC.
\end{definition}

As in the case of classical linear codes, the parameters of cLRCs satisfy fundamental trade-offs among length, dimension, minimum distance, and locality. The corresponding bounds provide the principal benchmarks for assessing the performance of LRC constructions. We summarize the bounds used in this work below.

\begin{lemma}\label{lemma:cLRC Bounds}
Let $\mathcal C$ be an $(n,k,d,r)_q$-cLRC, and let $k_{\mathrm{opt}}^{(q)}(n,d)$ denote the optimal dimension of a $q$-ary linear code of length $n$ and minimum distance $d$. Suppose that $\tau$ is a positive integer. Then the following bounds hold.
\begin{enumerate}
    \item \textbf{Singleton-like bound \cite{gopalan2012locality}}: We have that
    \begin{equation}\label{eq:classical-singleton}
    d\le n-k-\left\lceil\frac{k}{r}\right\rceil+2.
  \end{equation}

    \item \textbf{Cadambe-Mazumdar-like bound \cite{cadambe2015bounds}}: We have that
  \begin{equation}
   k\le \min_{\tau\in\mathbb{Z}_{\ge0}}\left\{\tau r+k_{\mathrm{opt}}^{(q)}\bigl(n-\tau(r+1),d\bigr)\right\},
\end{equation}

\item \textbf{Plotkin-like bound \cite{hao2020bounds}}: We have that
    \begin{equation}
    d \le \min_{0 \le \tau \le \lceil \frac{k}{r} \rceil - 1} \left\{ \frac{q^{k-\tau r-1}(q-1)(n-\tau(r+1))}{q^{k-\tau r} - 1} \right\}. \label{eq:plotkin}
    \end{equation}

\item \textbf{Sphere-packing-like bound \cite{li2025improved}}: We have that
    \begin{equation}\label{eq:sphere-packing}
    k \le n - \max_{0 \le \tau \le \lfloor \frac{n-1}{r+1} \rfloor} \left\{ \tau + \log_q \left( \sum_{i=0}^{\lfloor \frac{d-1}{2} \rfloor} \binom{n-\tau(r+1)}{i} (q-1)^i \right) \right\}. 
    \end{equation}

    \item \textbf{Griesmer-like bound \cite{hao2020bounds}}: We have that
    \begin{equation}
    n \ge \max_{0 \le \tau \le \lceil \frac{k}{r} \rceil - 1} \left\{ \tau(r+1) + \sum_{i=0}^{k-\tau r-1} \left\lceil \frac{d}{q^i} \right\rceil \right\}. \label{eq:griesmer}
    \end{equation}
\end{enumerate}
\end{lemma}

Two widely studied families of optimal locally recoverable codes (LRCs), namely Tamo--Barg (TB) codes and cyclic LRCs, are defined as follows.

\begin{definition}\cite{tamo2014family}
Let $\mathbb{F}_q$ be a finite field,  $r \ge 2$ be an integer such that $r+1$ divides $q-1$., and  $1\le \ell<q-1$. Let $\mathcal{P}=\mathbb{F}_q^{*}=\{1,\alpha,\alpha^2,\ldots,\alpha^{q-2}\}$,
where $\alpha$ is a primitive element of $\mathbb{F}_q$, so that the code length is $n=|\mathcal{P}|=q-1$. Define $S=\{\,i\in[\ell-1]^\dagger:\; i\not\equiv r \pmod {r+1}\,\}$.
The \emph{Tamo-Barg locally recoverable code} is the evaluation code
\[
C_{\mathrm{TB}}=\left\{
\operatorname{ev}(f)
=
\bigl(f(1),f(\alpha),\ldots,f(\alpha^{q-2})\bigr)
:\;
f(X)=\sum_{i\in S}a_iX^i,\;
a_i\in\mathbb{F}_q
\right\}.
\]
It is an $[n,k]_q$ locally recoverable code with $n=q-1,\; k=|S|=\ell-\left\lfloor\frac{\ell}{r+1}\right\rfloor,$ and locality $r$. Their dual exponent set 
\[
S^\perp = \{i \in [ q-2] : i \le q-\ell-1 \text{ or } i \equiv 1 \pmod{r+1}\}.
\]
The dual of Tamo-Barg code  given by
\[
C_{\mathrm{TB}}^\perp = \left\{ \operatorname{ev}(f) = \bigl(f(1), f(\alpha), \ldots, f(\alpha^{q-2})\bigr) : f(X) = \sum_{j \in S^\perp} a_j X^j, \; a_j \in \mathbb{F}_q \right\}.
\]
\end{definition}

For cyclic codes, locality can be characterized through the dual code. Luo \textit{et al.}~\cite{luo2018optimal} showed that a cyclic code $\mathcal{C}$ has locality $r$ if $\mathcal{C}^{\perp}$ contains a nonzero codeword of weight at most $r+1$. Building on this characterization, Tamo \textit{et al.}~\cite{tamo2015cyclic} derived a sufficient condition on the defining set that guarantees locality $r$ and yields optimal cyclic LRC constructions.
\begin{definition}[{\cite{tamo2015cyclic}}]\label{cyclicLRC}
Let $r+1$ divide $n$, $\alpha$ be a primitive $n$-th root of unity, and let $\mu=\frac{k}{r}$. Choose integers $\ell$ and $b$ such that $0\le \ell<r+1$ and $\gcd(b,n)=1$. Define $L=\{\alpha^i:\; i\equiv \ell \pmod{r+1}\}$, and $D=\{\alpha^{j+sb}:\; s=0,1,\ldots,n-\mu(r+1)\}$,
where $\alpha^j\in L$. Then the cyclic code with a defining set $Z=L\cup D$ is an optimal $(n,k,r)$ cyclic locally recoverable code.
\end{definition}

\subsection{Entanglement-Assisted Quantum Locally Recoverable Codes}
Let $\mathbb{C}$ denote the field of complex numbers, and let $\mathcal{H}=\mathbb{C}^q$ be the Hilbert space corresponding to a single qudit of dimension $q$. An $n$-qudit quantum system is described by the tensor product Hilbert space $\mathcal{H}^{\otimes n}\cong\mathbb{C}^{q^n}$, whose computational basis is given by $\{|x\rangle=|x_1\rangle\otimes\cdots\otimes|x_n\rangle:x=(x_1,\ldots,x_n)\in\mathbb{F}_q^n\}$. Accordingly, every pure quantum state can be expressed as $|\psi\rangle=\sum_{x\in\mathbb{F}_q^n}u_x|x\rangle$, where $u_x\in\mathbb{C}$ satisfy the  $\sum_{x\in\mathbb{F}_q^n}|u_x|^2=1$. Let $u=(u_1,\ldots,u_n)$ and $v=(v_1,\ldots,v_n)$ be vectors in $\mathbb{F}_q^n$, and let $\omega=e^{2\pi i/p}$, where $q=p^m$ for some prime $p$. The generalized shift and phase operators $X(u)$ and $Z(v)$ act on the computational basis as
\[
X(u)|x\rangle=|x+u\rangle,\qquad
Z(v)|x\rangle=\omega^{\operatorname{tr}(\langle v,x\rangle)}|x\rangle,
\]
where $\langle v,x\rangle=\sum_{i=1}^{n}v_ix_i$ is the standard Euclidean inner product over $\mathbb{F}_q$. The generalized Pauli group on $n$ qudits is defined by $\mathcal{G}_n=\{\omega^{\,j}X(u)Z(v):u,v\in\mathbb{F}_q^n,\;j\in\mathbb{F}_p\}$, and, for any error operator $E=\omega^{\,j}X(u)Z(v)\in\mathcal{G}_n$, its support is defined as $\operatorname{supp}(E)=\{\,i\in[n]:u_i\neq0\ \text{or}\ v_i\neq0\,\}$, which is the set of qudit positions on which $E$ acts nontrivially.

A quantum code $\mathcal{Q}_E$ is called an $[[n,\kappa,\delta;c]]_q$  entanglement-assisted quantum error-correcting code (EAQECC) if it encodes $\kappa$ logical qudits into $n$ physical qudits with the assistance of $c$ copies of a maximally entangled state, and has minimum distance $\delta$, meaning it can correct up to $\left\lfloor(\delta-1)/2\right\rfloor$ arbitrary quantum errors. In the stabilizer formalism, an EAQECC is specified by a stabilizer group 
\begin{equation}\label{eq:stb_EAQEC}
    \mathcal{S}_E=\mathcal{S}_S\times\mathcal{S}_I\subseteq\mathcal{G}_{n+c},
\end{equation}
where $\mathcal{S}_S = \langle g'_1, \dots, g'_c, h'_1, \dots, h'_c \rangle$ is generated by $c$ symplectic pairs and $\mathcal{S}_I = \langle g'_{c+1}, \dots, g'_{n-\kappa} \rangle$ by $n-\kappa-c$ commuting isotropic generators, and  the corresponding extended stabilizer generators are defined by
$g'_i=g_i^{\mathcal E}\otimes Z_i^{\mathcal D}, \;\text{for}\;1\le i\le c;\;
h'_i=h_i^{\mathcal E}\otimes X_i^{\mathcal D},\;\text{for}\; 1\le i\le c;\;
g'_j =g_j^{\mathcal E}\otimes I_{\mathcal D},\; \text{for}\;c+1\le j\le n-\kappa,$
where $g_i^{\mathcal{E}}, h_i^{\mathcal{E}}, g_j^{\mathcal{E}} \in \mathcal{G}_n$ and $Z^{\mathcal{D}}_{i},X^{\mathcal{D}}_{i}\in\mathcal{G}_c$. The corresponding code space $\mathcal{Q}_{S_E}$ is the simultaneous $+1$ eigenspace of $\mathcal{S}_E$ i.e., $\mathcal S_E$ be a set of $n-\kappa+c$ independent Pauli operators on $n+c$ qudits, generated by operators of the form $E=E_{\mathcal E}\otimes E_{\mathcal D}$, where $E_{\mathcal E}$ acts on the $n$ encoder (channel) qudits and $E_{\mathcal D}$ acts on the decoder's half of $c$ pre-shared maximally entangled pairs. Unlike an ordinary stabilizer group, the encoder-side components $E_{\mathcal E}$ need not pairwise commute; any anticommutation on the encoder side is compensated by the corresponding decoder-side operators, so that every generator of $\mathcal S_E$ commutes globally.

We briefly recall the notation for density operators. Let $\mathcal{H}_{\Omega}=(\mathbb{C}^{q})^{\otimes |\Omega|}$ denote the Hilbert space associated with a finite set $\Omega$ of qudits, and let $D(\Omega)$ denote the set of density operators on $\mathcal{H}_{\Omega}$. A density operator $\rho\in D(\Omega)$ is a positive semidefinite Hermitian operator with unit trace, i.e., $\operatorname{Tr}(\rho)=1$. A pure state $|\psi\rangle\in\mathcal{H}_{\Omega}$ is represented by the density operator $\rho=|\psi\rangle\langle\psi|$. For any subsystem $A\subseteq\Omega$, the corresponding reduced density operator is $\rho_A=\operatorname{Tr}_{\Omega\setminus A}(\rho)\in D(A)$,
where $\operatorname{Tr}_{\Omega\setminus A}$ denotes the partial trace over the complement of $A$.
 
 A quantum channel from a subsystem $\Omega$ to a subsystem $\Lambda$ is a completely positive trace-preserving (CPTP) map $\mathcal{N}:D(\Omega)\rightarrow D(\Lambda)$. By the Kraus representation, every quantum channel can be expressed as $\mathcal{N}(\rho)=\sum_k K_k\rho K_k^\dagger$, where the Kraus operators $K_k:\mathcal{H}_{\Omega}\rightarrow\mathcal{H}_{\Lambda}$ satisfy the completeness relation $\sum_k K_k^\dagger K_k=I$.
 
Recently, entanglement-assisted quantum locally recoverable codes (EA-qLRCs), introduced independently by Li \emph{et al.}~\cite{li2026entanglement} and Vijay \emph{et al.}~\cite{kumar2026entanglement}, generalize quantum locally recoverable codes by incorporating pre-shared maximally entangled pairs into the local recovery process. They are defined as follows.

\begin{definition}\label{EAqLRC}\cite{li2026entanglement,kumar2026entanglement}
An entanglement-assisted quantum code $\mathcal{Q}_E$
is called an \emph{entanglement-assisted quantum locally recoverable code} (EA-qLRC) with \emph{locality} $r$ if the following conditions are satisfied.
\begin{enumerate}
\item[(i)] For every coordinate $i\in[n]$, there exists a recovery set $\Gamma_i\subseteq[n]$ with $i\in\Gamma_i$ and 
$|\Gamma_i|\le r+1,$
and a completely positive trace-preserving (CPTP)  recovery channel
\[\mathcal R_i:
\mathsf D\!\left(
\mathcal H_{\Gamma_i\setminus\{i\}}
\otimes
\mathcal H_{\mathcal D}
\right)
\longrightarrow
\mathsf D\!\left(
\mathcal H_{\Gamma_i}
\right),\]
where $\mathcal H_{\mathcal D}$ denotes the decoder's share of the pre-shared maximally entangled state.
\item[(ii)] For every code state $\rho\in\mathcal Q_E,$
the erased qudit at coordinate $i$ can be recovered by applying $\mathcal R_i$ 
\[
\left(
\mathcal{R}_i
\otimes
I_{[n]\setminus\Gamma_i}
\right)
\left(
\rho_{([n]\setminus\{i\})}
\right)
=
\rho.
\]
\end{enumerate}
\end{definition}

Next,  stabilizer characterization of EA-qLRCs by incorporating the locality into the entanglement-assisted stabilizer formalism.

\begin{theorem}\label{Stb_EAqLRC}\cite{li2026entanglement,kumar2026entanglement}
Let $\mathcal{Q}_{\mathcal{S}_E}\subseteq(\mathbb{C}^q)^{\otimes(n+c)}$ be an
entanglement-assisted quantum stabilizer code with stabilizer group $\mathcal{S}_E$.
Suppose that for every encoder coordinate $i\in[n]$, there exist two stabilizer generators
\[
E_1 = E_{\mathcal{E}} \otimes E_{\mathcal{D}}, \qquad E_2 = E'_{\mathcal{E}} \otimes E'_{\mathcal{D}} \;\in\; \mathcal{S}_E,
\]
where $E_{\mathcal{E}} = X^\alpha Z^\beta$ and $E'_{\mathcal{E}} = X^{\alpha'} Z^{\beta'}$
satisfy $(\alpha_i,\beta_i)=(1,0)$ and $(\alpha_i',\beta_i')=(0,1)$. If the associated
recovery set $\Gamma_i = \operatorname{supp}(E_{\mathcal{E}})\cup\operatorname{supp}(E'_{\mathcal{E}})$
satisfies $|\Gamma_i|\le r+1$, then qudit $i$ can be reconstructed, following its erasure,
using only the surviving qudits in $\Gamma_i\setminus\{i\}$ together with the decoder's
shared ebits. Consequently, $\mathcal{Q}_{\mathcal{S}_E}$ is an EA-qLRC with locality $r$.
\end{theorem}

From Theorem~\ref{Stb_EAqLRC}, the following result establishes a CSS-type construction of EA-qLRCs from classical LRCs.

\begin{proposition}\label{CSS-EA-qLRC}\cite{li2026entanglement,kumar2026entanglement}
Let $\mathcal{C}_i$ be an $[n,k_i,d_i]_q$ linear code over $\mathbb{F}_q$ with
parity check matrix $H_i$ for $i=1,2$ and $c = \operatorname{rank}(H_1 H_2^{T})$. If for every coordinate $i \in [n]$, 
there exist parity-check vectors 
$c_1^{(i)} \in \mathcal{C}_1^\perp$ and $c_2^{(i)} \in \mathcal{C}_2^\perp$ such that
$i \in \operatorname{supp}(c_1^{(i)}) \cap \operatorname{supp}(c_2^{(i)}),$
and $\left| \operatorname{supp}(c_1^{(i)}) \cup \operatorname{supp}(c_2^{(i)})\right| \le r+1$, then $\mathcal{Q}_E$ is an  [$n,\kappa,\delta;c$]$_q$ EA-qLRC with locality $r$.
\end{proposition}

\subsection{Asymmetric Entanglement-Assisted Quantum Codes}
We now introduce the fundamental definition of asymmetric entanglement-assisted quantum error-correcting (AEAQEC) codes. This extends the standard EAQEC framework by allowing different error-correction capabilities for bit-flip $(X)$ and phase-flip $(Z)$ errors.

\begin{definition}
An asymmetric entanglement-assisted quantum error-correcting code (AEAQECC) with parameters $[[n,\kappa,\{d_X,d_Z\};c]]_q$ encodes $\kappa$ logical qudits into $n$ physical qudits using $c$ pre-shared maximally entangled pairs between the encoder and decoder, which can correct all bit-flip errors up to $\lfloor(d_X-1)/2\rfloor$ and all phase-flip errors up to $\lfloor(d_Z-1)/2\rfloor$.
\end{definition}

\begin{theorem}\cite{galindo2020asymmetric}\label{PC-AEAQEC}
Let $\mathcal{C}_i$ be an $[n,k_i,d_i]_q$ linear code over $\mathbb{F}_q$ with parity-check matrix $H_i$, for $i=1,2$. Then there exists an AEAQEC code $\mathcal{Q}_E$ with parameters $[[n,\,k_1+k_2-n+c,\{d_X,d_Z\};c]]_q$, where $d_X=\mathrm{wt}(\mathcal{C}_2\setminus(\mathcal{C}_1^\perp\cap\mathcal{C}_2))\;,d_Z=\mathrm{wt}(\mathcal{C}_1\setminus(\mathcal{C}_1\cap\mathcal{C}_2^\perp)$ and $c=\operatorname{rank}(H_1H_2^{T})=n-k_1-\dim_{\mathbb{F}_q}(\mathcal{C}_1^\perp\cap\mathcal{C}_2)=n-k_2-\dim_{\mathbb{F}_q}(\mathcal{C}_1\cap\mathcal{C}_2^\perp)$ is the minimum number of required maximally entangled pairs.
\end{theorem}

In particular, if $d_X=\mathrm{wt}(\mathcal{C}_2\setminus(\mathcal{C}_1^\perp\cap\mathcal{C}_2))=\mathrm{wt}(\mathcal{C}_2)$ and $d_Z=\mathrm{wt}(\mathcal{C}_1\setminus(\mathcal{C}_1\cap\mathcal{C}_2^\perp))=\mathrm{wt}(\mathcal{C}_1)$, then we say that the result AEAQEC code $\mathcal{Q}_E$ is a pure AEAQEC code.

The parameters of an AEAQEC code satisfy the following Singleton-like bound.
\begin{lemma}\cite{hu2024new}
Let $\mathcal{Q}_E$ be an $[[n,\,k_1+k_2-n+c,\{d_X,d_Z\};c]]_q$ AEAQEC code. Then
\[
d_X+d_Z \leq n-\kappa+c+2.
\]
\end{lemma}

\section{Asymmetric Entanglement-Assisted Quantum Locally Recoverable Codes}\label{sec:AEA-qLRCs}
In this work, we consider an asymmetric error model for local recovery, in which errors at known locations are restricted to be either $X$-type or $Z$-type, and we investigate the corresponding local recovery properties by analyzing the recovery mechanisms and locality requirements for each error type separately.

\begin{definition}[$X$-type Pauli Twirling Channel]
Let $\rho\in D(\mathcal{H}_{\Omega})$ be a density operator. The
\emph{$X$-type Pauli twirling channel} is the CPTP map
\begin{equation}
\mathcal{N}^{X}(\rho)=\frac{1}{q}\sum_{u\in\mathbb{F}_{q}} X^{u}\rho\left(X^{u}\right)^{\dagger}.
\end{equation}
where $\{X^{u}:u\in\mathbb{F}_{q}\}$ denotes the Pauli-$X$
operators on $\mathbb{C}^{q}$. For an $n$-qudit state
$\rho\in{D}(\mathcal{H}_{\Omega})$, the $X$-type Pauli twirling channel acting on the $i^{th}$ qudit is defined by
\[
\mathcal{N}^{X}_{i}(\rho)=\left(\mathrm{I}^{\otimes(i-1)}\otimes\mathcal{N}^{X}
\otimes
\mathrm{I}^{\otimes(n-i)}\right)\rho,
\]
where $\mathrm{I}$ denotes the identity channel.
\end{definition}

\begin{definition}[$Z$-type Pauli Twirling Channel]
Let $\rho\in{D}(\mathcal{H}_{\Omega})$ be a density operator. The
\emph{$Z$-type Pauli twirling channel} is the CPTP map
\[
\mathcal{N}^{Z}(\rho)=\frac{1}{q}\sum_{v\in\mathbb{F}_{q}}Z^{v}\rho\left(Z^{v}\right)^{\dagger},
\]
where $\{Z^{v}:v\in\mathbb{F}_{q}\}$ denotes the Pauli-$Z$
operators on $\mathbb{C}^{q}$. For an $n$-qudit state
$\rho\in{D}(\mathcal{H}_{\Omega})$,
the $Z$-type Pauli twirling channel acting on the $i^{th}$  qudit is defined by
\[
\mathcal{N}^{Z}_{i}(\rho)=\left(\mathrm{I}^{\otimes(i-1)}\otimes\mathcal{N}^{Z}\otimes
\mathrm{I}^{\otimes(n-i)}\right)\rho,
\]
where $\mathrm{I}$ denotes the identity channel.
\end{definition}

 \begin{definition}\label{EAqLRC-Asym-def}
An entanglement-assisted quantum code $\mathcal{Q}_E$ is called an asymmetric entanglement-assisted quantum locally recoverable code (AEA-qLRC) with locality $\{r_X, r_Z\}$, if for every coordinate $i\in[n]$, the following conditions are satisfied: 
\begin{enumerate}
    \item \textbf{Local recovery from $X$-type errors:} Suppose that only the $i^{th}$ qudit is affected by an $X$-type error. Then there exists a recovery set $\Gamma_i^X\subseteq[n]$ with $i\in\Gamma_i^X$ and $\vert{}\Gamma_i^X\vert{}\le r_X+1$, and a CPTP recovery channel:
    \[\mathcal R_i^X: \mathsf D\!\left( \mathcal H_{\Gamma_i^X\setminus\{i\}} \otimes \mathcal H_{\mathcal D} \right) \longrightarrow \mathsf D\!\left( \mathcal H_{\Gamma_i^X} \right),\]
    where $\mathcal H_{\mathcal D}$ denotes the decoder's share of the pre-shared maximally entangled state. For every code state $\rho\in\mathcal Q_E,$ the qudit at coordinate $i$ can be recovered by applying $\mathcal R_i^X$:
  \[\left( \mathcal{R}_i^X \otimes \mathrm{I}_{[n]\setminus\Gamma_i^X} \right) \left( \rho_{([n]\setminus\{i\})} \right) = \rho.\]
  
  \item  \textbf{Local recovery from $Z$-type errors:} Suppose that only the $i^{th}$ qudit is affected by an $Z$-type error. Then there exists a recovery set $\Gamma_i^Z\subseteq[n]$ with $i\in\Gamma_i^Z$ and $\vert{}\Gamma_i^Z\vert{}\le r_Z+1$, and a CPTP recovery channel:
  \[\mathcal R_i^Z: \mathsf D\!\left( \mathcal H_{\Gamma_i^Z\setminus\{i\}} \otimes \mathcal H_{\mathcal D} \right) \longrightarrow \mathsf D\!\left( \mathcal H_{\Gamma_i^Z} \right),\]
where $\mathcal H_{\mathcal D}$ denotes the decoder's share of the pre-shared maximally entangled state. For every code state $\rho\in\mathcal Q_E,$ the qudit at coordinate $i$ can be recovered by applying $\mathcal R_i^Z$:
  \[\left( \mathcal{R}_i^Z \otimes \mathrm{I}_{[n]\setminus\Gamma_i^Z} \right) \left( \rho_{([n]\setminus\{i\})} \right) = \rho.\]
\end{enumerate}
 \end{definition}

The stabilizer framework in \eqref{eq:stb_EAQEC} extends naturally to the AEAQEC setting. The following Theorem~\ref{Stb_EAqLRC_Asym} establishes a stabilizer sufficient condition of AEA-qLRCs by incorporating the asymmetric locality requirements into the entanglement-assisted stabilizer formalism.

\begin{theorem}\label{Stb_EAqLRC_Asym}
Let $\mathcal{Q}_{\mathcal{S}_E}$ be an entanglement-assisted quantum stabilizer code with stabilizer group $\mathcal{S}_E$. Suppose that for every encoder coordinate $i\in[n]$, there exist two stabilizer generators $E_1 = E_{\mathcal{E}} \otimes E_{\mathcal{D}}$ and $E_2 = E'_{\mathcal{E}} \otimes E'_{\mathcal{D}}$, where $E_{\mathcal{E}} = X^\alpha Z^\beta$ and $E'_{\mathcal{E}} = X^{\alpha'} Z^{\beta'}$ such that $(\alpha_i,\beta_i)=(1,0)$ and $(\alpha_i',\beta_i')=(0,1)$. If the associated recovery sets $\Gamma_i^Z = \operatorname{supp}(E_{\mathcal E})$ and $\Gamma_i^X = \operatorname{supp}(E_{\mathcal E}')$ satisfy $\vert{}\Gamma_i^Z\vert{}\le r_Z+1$ and $\vert{}\Gamma_i^X\vert{}\le r_X+1$, then every X-type error acting on the $i$-th encoded qudit can be recovered by accessing at most $r_X$ other encoded qudits, and every Z-type error acting on the $i$-th encoded qudit can be recovered by accessing at most $r_Z$ other encoded qudits. Consequently, $\mathcal{Q}_{\mathcal{S}_E}$ is an AEA-qLRC with  locality $\{r_X, r_Z\}$.
\end{theorem} 
\begin{proof}
We first make explicit the reduction from qudit \emph{erasure} to a
\emph{known-location Pauli error}, which underlies the syndrome-based recovery
argument below. When qudit $i$ is erased, the decoder cannot access its original
contents; to reconstruct it, the recovery channel conceptually substitutes a fixed
reference state $\lvert 0\rangle$ at position $i$ and treats the discrepancy between
this substitution and the true (unknown) state of qudit $i$ as an unknown Pauli
error $E$ supported only on coordinate $i$. Concretely, if the original code state
had reduced density operator $\rho_i$ on qudit $i$ satisfying
$X_i^{u}\lvert 0\rangle\!\langle 0\rvert (X_i^{u})^{\dagger}=\rho_i$ for some
$u\in\mathbb{F}_q$, then recovering qudit $i$ amounts to identifying $u$ and applying
$X_i^{u}$ to the substituted reference qudit; this is exactly the role played by the
error operator $E=X^{e_X}\otimes I_{\mathcal D}$ (respectively $Z^{e_Z}\otimes I_{\mathcal D}$)
introduced below, with $u=(e_X)_i$ (respectively $(e_Z)_i$) the unknown quantity to be
determined. Crucially, the syndrome that identifies $u$ is extracted by measuring a
stabilizer element whose \emph{encoder-side support outside coordinate $i$} lies
entirely within the recovery set $\Gamma_i^X\setminus\{i\}$ (respectively
$\Gamma_i^Z\setminus\{i\}$); since a stabilizer measurement outcome is determined by
the code state and does not depend on the (as yet unset) reference qudit at $i$, this
measurement is well-defined using only the surviving qudits in
$\Gamma_i^X\setminus\{i\}$ together with the decoder's ebits. Once the syndrome is
obtained, applying the corresponding correction to the reference qudit at position
$i$ completes the reconstruction. This two-step procedure — syndrome extraction from
$\Gamma_i^X\setminus\{i\}\cup\mathcal{H}_{\mathcal D}$, followed by correction of the
reference qudit at $i$ — realizes precisely the recovery channel
$\mathcal R_i^X:\mathsf D(\mathcal H_{\Gamma_i^X\setminus\{i\}}\otimes\mathcal
H_{\mathcal D})\to\mathsf D(\mathcal H_{\Gamma_i^X})$ (respectively $\mathcal R_i^Z$)
required by Definition~\ref{EAqLRC-Asym-def}. We now carry out this procedure explicitly.

\smallskip
Fix a coordinate $i\in[n]$ of the encoder's $n$ physical qudits. By assumption, there exist two stabilizer generators $E_1=E_{\mathcal E}\otimes E_{\mathcal D}$ and $E_2=E_{\mathcal E}'\otimes E_{\mathcal D}'$ where $E_{\mathcal E}=X^\alpha Z^\beta$ and $E_{\mathcal E}'=X^{\alpha'}Z^{\beta'}$ such that $(\alpha_i,\beta_i)=(1,0)$ and $(\alpha_i',\beta_i')=(0,1)$. Let $\Gamma_i^Z=\operatorname{supp}(E_{\mathcal E})$ and $\Gamma_i^X=\operatorname{supp}(E_{\mathcal E}')$ denote the corresponding recovery sets. We show that any restricted X-type (respectively, Z-type) error acting on the $i$-th encoded qudit can be corrected by accessing only the qudits in $\Gamma_i^X\setminus\{i\}$ (respectively, $\Gamma_i^Z\setminus\{i\}$) together with the decoder's ebits.

 Local recovery from X-type errors: Suppose an X-type error occurs on the $i$-th qudit. Let $e=(e_X\mid 0)\in\mathbb F_q^{2n}$ be the symplectic representation of the error, where $\operatorname{supp}(e_X)=\{i\}$. Equivalently, the corresponding error operator is $E=\left(X^{e_X}\right)\otimes I_{\mathcal D}$.
To identify the error, the recovery channel performs syndrome measurements using the commuting family $M_{E_2}=\{E_2^\eta:\eta\in\mathbb F_q\}$. Since $[I_{\mathcal D},{E'}_{\mathcal D}^{\,\eta}]=0$, the decoder subsystem contributes no phase. The measured syndrome depends only on the symplectic inner product over the encoder subsystem:
  $$\lambda_{E_2}(\eta)=\omega^{\operatorname{tr}\!\left(\langle\eta(\alpha'\mid\beta'),(e_X\mid 0)\rangle_s\right)}=\omega^{\operatorname{tr}\!\left(\eta(e_X\cdot\beta')\right)}.$$
Accordingly, the extracted syndrome is $s_{E_2}=e_X\cdot\beta'$. Since the error is supported only on coordinate $i$, this reduces to $s_{E_2}=(e_X)_i\beta_i'$. Using $(\alpha_i',\beta_i')=(0,1)$, we obtain $s_{E_2}=(e_X)_i(1)=(e_X)_i$.
Therefore, the recovery operator is:
  $$R_i^X=\left(X_i^{-s_{E_2}}\right)\otimes I_{\mathcal D}=\left(X_i^{-(e_X)_i}\right)\otimes I_{\mathcal D}.$$Applying $R_i^X$ to the error operator yields:$$R_i^XE=\left(X_i^{-(e_X)_i}\otimes I_{\mathcal D}\right)\left(X_i^{(e_X)_i}\otimes I_{\mathcal D}\right)\propto I_{\mathcal E}\otimes I_{\mathcal D}.$$
Hence, the X-type error is corrected exactly. Finally, $\vert{}\Gamma_i^X\setminus\{i\}\vert{}=\vert{}\Gamma_i^X\vert{}-1\le(r_X+1)-1=r_X$. So the recovery operation accesses at most $r_X$ additional encoded qudits.

Local recovery from Z-type errors: 
Suppose a Z-type error occurs on the $i$-th qudit. Let $e=(0\mid e_Z)\in\mathbb F_q^{2n}$ be the symplectic representation of the error, where $\operatorname{supp}(e_Z)=\{i\}$. Equivalently, the corresponding error operator is $E=\left(Z^{e_Z}\right)\otimes I_{\mathcal D}$.
To identify the error, the recovery channel performs syndrome measurements using the commuting family $M_{E_1}=\{E_1^\eta:\eta\in\mathbb F_q\}$. Since $[I_{\mathcal D},E_{\mathcal D}^{\,\eta}]=0$, the measured syndrome is:
  $$\lambda_{E_1}(\eta)=\omega^{\operatorname{tr}\!\left(\langle\eta(\alpha\mid\beta),(0\mid e_Z)\rangle_s\right)}=\omega^{\operatorname{tr}\!\left(\eta(-e_Z\cdot\alpha)\right)}.$$Accordingly, the extracted syndrome is $s_{E_1}=-e_Z\cdot\alpha$. Since the error is supported only on coordinate $i$, this reduces to $s_{E_1}=-(e_Z)_i\alpha_i$. Using $(\alpha_i,\beta_i)=(1,0)$, we obtain $s_{E_1}=-(e_Z)_i(1)=-(e_Z)_i$.Therefore, the recovery operator is:$$R_i^Z=\left(Z_i^{\,s_{E_1}}\right)\otimes I_{\mathcal D}=\left(Z_i^{-(e_Z)_i}\right)\otimes I_{\mathcal D}.$$Applying $R_i^Z$ to the error operator yields:$$R_i^ZE=\left(Z_i^{-(e_Z)_i}\otimes I_{\mathcal D}\right)\left(Z_i^{(e_Z)_i}\otimes I_{\mathcal D}\right)\propto I_{\mathcal E}\otimes I_{\mathcal D}.$$
Hence, the Z-type error is corrected exactly. Finally, $\vert{}\Gamma_i^Z\setminus\{i\}\vert{}=\vert{}\Gamma_i^Z\vert{}-1\le(r_Z+1)-1=r_Z$. So the recovery operation accesses at most $r_Z$ additional encoded qudits.  \\ \smallskip The simultaneous measurements on the decoder-side ebits do not contribute to the locality cost, since the shared entanglement is assumed to be freely available during decoding. Therefore, every coordinate admits a recovery procedure involving at most $ r_X$ other encoded qudits for X-type errors and at most $ r_Z$ other encoded qudits for Z-type errors, establishing that $ \mathcal {Q}_{S_E}$ is an AEA-qLRC with locality $\{r_X, r_Z\}$.
\end{proof}

As a consequence of Theorem~\ref{Stb_EAqLRC_Asym}, we obtain the following CSS-type construction of AEA-qLRCs from classical locally recoverable codes.

\begin{proposition}\label{CSS-EA-qLRC-Asym}Let $\mathcal{C}_i$ be an $[n,k_i,d_i]_q$ linear code over $\mathbb{F}_q$ with parity check matrix $H_i$ for $i=1,2$ and $c = \operatorname{rank}(H_1 H_2^{T})$. If for every coordinate $i \in [n]$, there exist parity-check vectors $c_1^{(i)} \in \mathcal{C}_1^\perp$ and $c_2^{(i)} \in \mathcal{C}_2^\perp$ such that $i \in \operatorname{supp}(c_1^{(i)}) \cap \operatorname{supp}(c_2^{(i)}),$ with $\left\vert{} \operatorname{supp}(c_1^{(i)}) \right\vert{} \le r_Z+1$ and $\left\vert{} \operatorname{supp}(c_2^{(i)}) \right\vert{} \le r_X+1$, then $\mathcal{Q}_E$ is an  $[[n,\kappa,\{d_X,d_Z\},\{r_X,r_Z\};c]]_q$ AEA-qLRC with locality $\{r_X, r_Z\}$.\end{proposition}

 \begin{proof}Let $\mathcal{Q}_E$ be the entanglement-assisted quantum code obtained from the classical linear codes $\mathcal{C}_1$ and $\mathcal{C}_2$ via the EA-CSS construction. By construction, the encoder-side components of the stabilizer generators are determined by the classical dual codes $\mathcal{C}_1^\perp=\operatorname{rowspan}(H_1)$ and $\mathcal{C}_2^\perp=\operatorname{rowspan}(H_2)$. Consequently, every encoder-side stabilizer operator has the form $E_{\mathcal E}=X^uZ^v,$ where $u\in\mathcal{C}_1^\perp$ and $v\in\mathcal{C}_2^\perp$.Assume that for every coordinate $i\in[n]$, there exist parity-check vectors $c_1^{(i)}\in\mathcal{C}_1^\perp$ and $c_2^{(i)}\in\mathcal{C}_2^\perp$ such that $i\in \operatorname{supp}(c_1^{(i)}) \cap \operatorname{supp}(c_2^{(i)}),$ with the asymmetric locality conditions $\left\vert{} \operatorname{supp}(c_1^{(i)}) \right\vert{} \le r_Z+1$ and $\left\vert{} \operatorname{supp}(c_2^{(i)}) \right\vert{} \le r_X+1$.  Since $i\in\operatorname{supp}(c_1^{(i)})$, we have $(c_1^{(i)})_i\neq0$. Multiplying $c_1^{(i)}$ by the nonzero scalar $((c_1^{(i)})_i)^{-1}$, we may assume without loss of generality that $(c_1^{(i)})_i=1$. The corresponding encoder-side stabilizer operator is $E_{\mathcal E}=X^{c_1^{(i)}}Z^0,$ so that $(\alpha_i,\beta_i)=(1,0)$. This operator allows for local recovery from Z-type errors. The recovery set size is bounded by $\left\vert{} \operatorname{supp}(E_{\mathcal E}) \right\vert{} = \left\vert{} \operatorname{supp}(c_1^{(i)}) \right\vert{} \le r_Z+1$.  Similarly, after normalizing $c_2^{(i)}$, we obtain $E'_{\mathcal E}=X^0Z^{c_2^{(i)}},$ which satisfies $(\alpha_i',\beta_i')=(0,1)$. This operator provides local recovery from X-type errors. The recovery set size is bounded by $\left\vert{} \operatorname{supp}(E'_{\mathcal E}) \right\vert{} = \left\vert{} \operatorname{supp}(c_2^{(i)}) \right\vert{} \le r_X+1$.  Because the stabilizer group contains these operators bounded independently by $r_Z+1$ and $r_X+1$ for every coordinate $i$, the stabilizer criterion for asymmetric locality $\{r_X, r_Z\}$ is satisfied. Therefore, $\mathcal Q_E$ is an AEA-qLRC with asymmetric locality $\{r_X, r_Z\}$.
\end{proof}

\section{Bounds of CSS-type AEA-qLRC}\label{sec:bounds}
Throughout this section, all bounds are stated for AEA-qLRCs arising from the CSS-type construction of Proposition~\ref{CSS-EA-qLRC-Asym}; we do not address
bounds for general (non-CSS) entanglement-assisted stabilizer codes with asymmetric locality, which remains open (see Section~\ref{sec:conclusion}).

In this section, we establish a Singleton-like bound for AEA-qLRCs with asymmetric localities $\{r_X,r_Z\}$, extending the corresponding bounds for EA qLRCs ~\cite{li2026entanglement,kumar2026entanglement}. Let $n_{\text{opt}}^{(q)}(k, d;r)$, $k_{\text{opt}}^{(q)}(n, d;r)$, and $d_{\text{opt}}^{(q)}(n, k;r)$ denote, respectively, the smallest length, the largest dimension, and the largest minimum distance of a $q$-ary linear code with the two remaining parameters fixed.

\begin{theorem}\label{Bounds_theorem_Asym}
Let $\mathcal{Q}_E$ be an AEA-qLRC with parameters
$[[n,\kappa,\{d_X,d_Z\},\{r_X,r_Z\};c]]_q$, constructed from
Proposition~\ref{CSS-EA-qLRC-Asym} using two classical
$[n,k_i]_q$ linear codes $\mathcal{C}_i$, $i=1,2$, where
$\mathcal{C}_1$ has distance $d_Z$ and locality $r_Z$, while
$\mathcal{C}_2$ has distance $d_X$ and locality $r_X$. Then
  $$\begin{aligned} d_X+d_Z &\leq d_{\mathrm{opt}}^q(k_1+c,\kappa;r_Z) + d_{\mathrm{opt}}^q(k_2+c,\kappa;r_X),\\ 2\kappa &\leq k_{\mathrm{opt}}^q(k_1+c,d_Z;r_Z) + k_{\mathrm{opt}}^q(k_2+c,d_X;r_X),\\ n+\kappa+c &\geq n_{\mathrm{opt}}^q(\kappa,d_X;r_X) + n_{\mathrm{opt}}^q(\kappa,d_Z;r_Z).  \end{aligned}$$
  \end{theorem}

  \begin{proof} Let $\mathcal{C}_i$ be classical linear codes over $\mathbb{F}_q$ for $i=1,2$. Let $\mathcal{Q}_E$ be an AEAQEC code constructed by Theorem~\ref{PC-AEAQEC}. Suppose that $\mathcal{C}_1^\perp \cap \mathcal{C}_2$ is classical linear code with  generator matrix $\begin{pmatrix} I_{n-k_1-c} & A \end{pmatrix}$.Then, we can assume that $\mathcal{C}_2$ has generator matrix 
\[G_{\mathcal{C}_2} = \begin{pmatrix} I_{n-k_1-c} & A_1 \\ O_{\kappa\times(n-k_1-c)} & P_1\end{pmatrix}.\]
The matrix $\begin{pmatrix} O_{\kappa\times(n-k_1-c)} & P_1\end{pmatrix}$ generates a code $\mathcal{C}'$ with parameters $[n,\kappa,d']_q$.
Every nonzero codeword of $\mathcal{C}'$ lies in $\mathcal{C}_2 \setminus (\mathcal{C}_1^\perp \cap \mathcal{C}_2)$, so $d' \geq d_X$. Let $\mathcal{C}_{P_1}$ be the linear code generated by $P_1$, which has parameters $[k_1+c,\,\kappa,\,d(\mathcal{C}_{P_1})\ge d_X]_q$.
By Proposition~\ref{CSS-EA-qLRC-Asym}, $\mathcal{C}_2$ has locality $r_X$, since $\mathcal{C}'\subseteq\mathcal{C}_2$, the subcode $\mathcal{C}'$ locality at most $r_X$, and hence $\mathcal{C}_{P_1}$ is a LRC with locality at most $r_X$. Therefore,
\begin{equation}\label{equ_3_asym}
d_X \leq d_{\mathrm{opt}}^q(k_1+c,\kappa;r_X), \qquad \kappa \leq k_{\mathrm{opt}}^q(k_1+c,d_X;r_X),\quad k_1+c \geq n_{\mathrm{opt}}^q(\kappa,d_X;r_X). \end{equation}
   \smallskip
 Let $\begin{pmatrix} I_{n-k_2-c} & A \end{pmatrix}$ be the generator matrix of code $\mathcal{C}_1\cap\mathcal{C}_2^{\perp}$.Then, we can assume that $\mathcal{C}_1$ has generator matrix 
\[G_{\mathcal{C}_1} = \begin{pmatrix} I_{n-k_2-c} & A_2 \\ O_{\kappa\times(n-k_2-c)} & P_2\end{pmatrix}.\]
Applying the similar argument, we obtain a code $\mathcal{C}_{P_2}$ with parameters $[k_2+c,\,\kappa,\,d(\mathcal{C}_{P_2})\ge d_Z]_q$. By Proposition~\ref{CSS-EA-qLRC-Asym}, $\mathcal{C}_1$ has locality $r_Z$, so $\mathcal{C}_{P_2}$ is a LRC with locality at most $r_Z$. Therefore,

\begin{equation}\label{equ_4_asym}
    d_Z \leq d_{\mathrm{opt}}^q(k_2+c,\kappa;r_Z), \qquad \kappa \leq k_{\mathrm{opt}}^q(k_2+c,d_Z;r_Z),\quad k_2+c \geq n_{\mathrm{opt}}^q(\kappa,d_Z;r_Z). 
\end{equation}
\smallskip
Adding~\eqref{equ_3_asym} and~\eqref{equ_4_asym} yields the first two claimed inequalities and $$(k_1+c)+(k_2+c) \geq n_{\mathrm{opt}}^q(\kappa,d_X;r_X) + n_{\mathrm{opt}}^q(\kappa,d_Z;r_Z).$$Substituting $k_1+k_2+2c= n+\kappa+c$ yields$$n+\kappa+c \geq n_{\mathrm{opt}}^q(\kappa,d_X;r_X) + n_{\mathrm{opt}}^q(\kappa,d_Z;r_Z),$$which is the third claimed inequality.
\end{proof}

\begin{corollary}\label{cor:singleton_Asym_direct}
Let $\mathcal{Q}_E$ be an AEA-qLRC with parameters
$[[n,\kappa,\{d_X,d_Z\},\{r_X,r_Z\};c]]_q$, constructed using
Proposition~\ref{CSS-EA-qLRC-Asym} from two classical
$[n,k_i]_q$ linear codes $\mathcal{C}_i$, $i=1,2$, where
$\mathcal{C}_1$ has distance $d_Z$ and locality $r_Z$, while
$\mathcal{C}_2$ has distance $d_X$ and locality $r_X$. Then
\begin{equation}\label{eqn:Singleton Bound AEA-qLRC}
    d_X+d_Z
\leq
n-\kappa+c
-\left\lceil\frac{\kappa}{r_X}\right\rceil
-\left\lceil\frac{\kappa}{r_Z}\right\rceil
+4.
\end{equation}
\end{corollary}
\begin{proof}
Applying the classical LRC Singleton bound~\eqref{eq:classical-singleton}
to the two terms in the first inequality of
Theorem~\ref{Bounds_theorem_Asym}, we obtain
\[
d_X+d_Z
\leq
d_{\mathrm{opt}}^q(k_1+c,\kappa;r_X)
+
d_{\mathrm{opt}}^q(k_2+c,\kappa;r_Z)
\leq
\left(
k_1+c-\kappa
-\left\lceil\frac{\kappa}{r_X}\right\rceil
+2
\right)
+
\left(
k_2+c-\kappa
-\left\lceil\frac{\kappa}{r_Z}\right\rceil
+2
\right),
\]
which simplifies to
\[
d_X+d_Z
\leq
k_1+k_2+2c-2\kappa
-\left\lceil\frac{\kappa}{r_X}\right\rceil
-\left\lceil\frac{\kappa}{r_Z}\right\rceil
+4.
\]
Since $k_1+k_2=n+\kappa-c$, substituting yields
\[
d_X+d_Z
\leq
n+\kappa-c+2c-2\kappa
-\left\lceil\frac{\kappa}{r_X}\right\rceil
-\left\lceil\frac{\kappa}{r_Z}\right\rceil
+4
=
n-\kappa+c
-\left\lceil\frac{\kappa}{r_X}\right\rceil
-\left\lceil\frac{\kappa}{r_Z}\right\rceil
+4,
\]
which proves the result.
\end{proof}

\begin{corollary}\label{cor:singleton_Asym}
Let $\mathcal{Q}_E$ be an AEA-qLRC with parameters
$[[n,\kappa,\{d_X,d_Z\},\{r_X,r_Z\};c]]_q$, constructed from
Proposition~\ref{CSS-EA-qLRC-Asym} using two classical
$[n,k_i]_q$ linear codes $\mathcal{C}_i$, $i=1,2$, where
$\mathcal{C}_1$ has distance $d_Z$ and locality $r_Z$, while
$\mathcal{C}_2$ has distance $d_X$ and locality $r_X$. Then
\begin{enumerate}\item \textbf{Cadambe-Mazumdar-like bound:} We have that$$\begin{aligned} 2\kappa\leq\;& \min_{\tau_1,\tau_2\in\mathbb{Z}_{\geq0}} \Big\{\tau_1r_X+\tau_2r_Z +k_{\mathrm{opt}}^{(q)} \!\left(k_2+c-\tau_1(r_X+1),d_X\right)\\ &\qquad\qquad+ k_{\mathrm{opt}}^{(q)} \!\left(k_1+c-\tau_2(r_Z+1),d_Z\right) \Big\}. \end{aligned}$$\item \textbf{Plotkin-like Bound:} We have that$$\begin{aligned} d_X+d_Z\leq &(q-1) \min_{\substack{0\leq\tau_1\leq\lceil\kappa/r_X\rceil-1\\ 0\leq\tau_2\leq\lceil\kappa/r_Z\rceil-1}} \Bigg\{ \frac{q^{\kappa-\tau_1r_X-1} (k_2+c-\tau_1(r_X+1))} {q^{\kappa-\tau_1r_X}-1}\\ &\qquad+ \frac{q^{\kappa-\tau_2r_Z-1} (k_1+c-\tau_2(r_Z+1))} {q^{\kappa-\tau_2r_Z}-1} \Bigg\}. \end{aligned}$$\item \textbf{Sphere-packing-like Bound:} We have that$$\begin{aligned} \kappa\leq\;& n+c- \max_{\substack{ 0\leq\tau_1\leq\lfloor(k_2+c-1)/(r_X+1)\rfloor\\ 0\leq\tau_2\leq\lfloor(k_1+c-1)/(r_Z+1)\rfloor}} \Bigg\{ \tau_1+\tau_2\\ &\qquad+ \log_q\!\Big[ V_q\!\big(k_2+c-\tau_1(r_X+1),t_X\big) V_q\!\big(k_1+c-\tau_2(r_Z+1),t_Z\big) \Big] \Bigg\}. \end{aligned}$$where $t_X = \lfloor \frac{d_X-1}{2} \rfloor$, $t_Z = \lfloor \frac{d_Z-1}{2} \rfloor$, and $V_q(M, t) = \sum_{j=0}^{t} \binom{M}{j}(q-1)^j$ is the volume of the classical $q$-ary Hamming ball.\item \textbf{Griesmer-like Bound:} We have that$$\begin{aligned} n+\kappa+c\geq\;& \max_{0\leq\tau_1\leq\lceil\kappa/r_X\rceil-1} \left\{\tau_1(r_X+1)+ \sum_{i=0}^{\kappa-\tau_1r_X-1} \left\lceil\frac{d_X}{q^i}\right\rceil\right\}\\ &+ \max_{0\leq\tau_2\leq\lceil\kappa/r_Z\rceil-1} \left\{\tau_2(r_Z+1)+ \sum_{i=0}^{\kappa-\tau_2r_Z-1} \left\lceil\frac{d_Z}{q^i}\right\rceil\right\}. \end{aligned}$$\end{enumerate}
\end{corollary}
\begin{proof}
By substituting the classical LRC bounds from Lemma~\ref{lemma:cLRC Bounds} into the respective inequalities of Theorem~\ref{Bounds_theorem_Asym}. For the Cadambe-Mazumdar, Plotkin, and Griesmer-like bounds, the results are obtained directly by evaluating the bound for $\mathcal{C}_1$ using parameters $(d_Z, r_Z)$ at evaluated length $k_1+c$, and $\mathcal{C}_2$ using parameters $(d_X, r_X)$ at evaluated length $k_2+c$. For the sphere-packing-like bound, applying the classical bound to $k_{\mathrm{opt}}^q(k_1+c,d_Z;r_Z)$ and $k_{\mathrm{opt}}^q(k_2+c,d_X;r_X)$ yields an upper bound on $2\kappa$. Substituting the  dimension relation $k_1 + k_2 = n + \kappa - c$ into this sum and subtracting $\kappa$ from both sides simplifies exactly to the stated inequality.\end{proof}

We next characterize the necessary and sufficient conditions for a pure AEA-qLRC to attain the Singleton-like bound \eqref{eqn:Singleton Bound AEA-qLRC}.
\begin{theorem}\label{pure_EA-qLRC-Asym}
Let $\mathcal{Q}_E$ be a pure AEA-qLRC with parameters
$[[n,\kappa,\{d_X,d_Z\},\{r_X,r_Z\};c]]_q$, constructed from
Proposition~\ref{CSS-EA-qLRC-Asym} using two classical
$[n,k_i]_q$ linear codes $\mathcal{C}_i$, $i=1,2$, where
$\mathcal{C}_1$ has distance $d_Z$ and locality $r_Z$, while
$\mathcal{C}_2$ has distance $d_X$ and locality $r_X$. Then $\mathcal{Q}_E$ is optimal if and only if $\mathcal{C}_1$ and $\mathcal{C}_2$ are optimal LRCs such that $d_Z=d_X$, $$\left\lceil\dfrac{k_1}{r_Z}\right\rceil = \left\lceil\dfrac{\kappa}{r_Z}\right\rceil \quad \text{and} \quad \left\lceil\dfrac{k_2}{r_X}\right\rceil = \left\lceil\dfrac{\kappa}{r_X}\right\rceil.$$
\end{theorem}

\begin{proof}$(\Leftarrow)$ Suppose $\mathcal{C}_1$ and $\mathcal{C}_2$ are optimal classical LRCs with $d_Z = d_X = d$, $\lceil k_1/r_Z \rceil = \lceil \kappa/r_Z \rceil$, and $\lceil k_2/r_X \rceil = \lceil \kappa/r_X \rceil$. From the classical LRC Singleton bound~\eqref{eq:classical-singleton}, we have $d_Z = n - k_1 - \lceil k_1/r_Z \rceil + 2$ and $d_X= n - k_2 - \lceil k_2/r_X \rceil + 2$. Thus,
$$\begin{aligned} 2d = d_Z+d_X &= 2n - (k_1+k_2) - \left\lceil\frac{k_1}{r_Z}\right\rceil - \left\lceil\frac{k_2}{r_X}\right\rceil + 4 \\ &= n - \kappa + c - \left\lceil\frac{\kappa}{r_Z}\right\rceil - \left\lceil\frac{\kappa}{r_X}\right\rceil + 4, \end{aligned}$$
since $\kappa = k_1+k_2 - n + c$. Because $\min\{d_X, d_Z\}=d$, $\mathcal{Q}_E$ attains the Singleton-like bound \eqref{eqn:Singleton Bound AEA-qLRC}, hence $\mathcal{Q}_E$ is optimal.

$(\Rightarrow)$ Suppose that $\mathcal{Q}_E$ is an $[[n,\kappa,\{d_X,d_Z\},\{r_X,r_Z\};c]]_q$ optimal pure AEA-qLRC. Then, from \eqref{eqn:Singleton Bound AEA-qLRC} $$2d = n - \kappa + c - \left\lceil\frac{\kappa}{r_Z}\right\rceil - \left\lceil\frac{\kappa}{r_X}\right\rceil + 4.$$Since $\mathcal{C}_1$ and $\mathcal{C}_2$ have localities $r_Z$ and $r_X$ respectively, the bound in~\eqref{eq:classical-singleton} gives $d_Z \le n - k_1 - \lceil k_1/r_Z \rceil + 2$ and $d_X \le n - k_2 - \lceil k_2/r_X \rceil + 2$. Adding these two inequalities, we get:$$2d \leq d_Z+d_X \le n - \kappa + c - \left\lceil\frac{k_1}{r_Z}\right\rceil - \left\lceil\frac{k_2}{r_X}\right\rceil + 4.$$Comparing this to the optimality equation implies:$$\left\lceil\frac{k_1}{r_Z}\right\rceil + \left\lceil\frac{k_2}{r_X}\right\rceil \le \left\lceil\frac{\kappa}{r_Z}\right\rceil + \left\lceil\frac{\kappa}{r_X}\right\rceil.$$Since $k_i \ge \kappa$, we inherently have $\lceil k_1/r_Z \rceil \ge \lceil \kappa/r_Z \rceil$ and $\lceil k_2/r_X \rceil \ge \lceil \kappa/r_X \rceil$. Thus, the inequality forces exactly:$$\left\lceil\frac{k_1}{r_Z}\right\rceil = \left\lceil\frac{\kappa}{r_Z}\right\rceil \quad \text{and} \quad \left\lceil\frac{k_2}{r_X}\right\rceil = \left\lceil\frac{\kappa}{r_X}\right\rceil.$$This in turn implies that $2d = \Bigl(n - k_1 - \lceil k_1/r_Z \rceil + 2\Bigr) + \Bigl(n - k_2 - \lceil k_2/r_X \rceil + 2\Bigr)$. Therefore, $d_Z = n - k_1 - \lceil k_1/r_Z \rceil + 2$ and $d_X = n - k_2 - \lceil k_2/r_X \rceil + 2$, meaning both $\mathcal{C}_1$ and $\mathcal{C}_2$ independently attain the classical bound \eqref{eq:classical-singleton}. Hence, $\mathcal{C}_1$ and $\mathcal{C}_2$ are optimal LRCs.\end{proof}

\begin{lemma}\label{lem:ceiling}
Let $k,r$ be positive integers and let $s$ be an integer with $0\le s\le k$. Then
\[
\left\lceil\frac{k}{r}\right\rceil=\left\lceil\frac{k-s}{r}\right\rceil
\quad\Longleftrightarrow\quad
s\le (k-1)\bmod r.
\]
\end{lemma}
\begin{proof}
Write $k=ar+b$ with $a=\lfloor k/r\rfloor$ and $0\le b<r$.

\emph{Case $b=0$.} Then $\lceil k/r\rceil=a$, and $k-s=(a-1)r+(r-s)$. If
$0\le s\le r-1$, then $0<r-s\le r$, so $\lceil (k-s)/r\rceil=a$; if $s\ge r$, then
$\lceil(k-s)/r\rceil\le a-1$. Hence $\lceil k/r\rceil=\lceil(k-s)/r\rceil$ if and only
if $s\le r-1$. Since $b=0$, we have $(k-1)\bmod r=r-1$, so this is exactly the
condition $s\le(k-1)\bmod r$.

\emph{Case $1\le b\le r-1$.} Then $\lceil k/r\rceil=a+1$, and $k-s=ar+(b-s)$. If
$s<b$, then $0<b-s\le r-1$, so $\lceil(k-s)/r\rceil=a+1$; if $s\ge b$, then
$\lceil(k-s)/r\rceil\le a$. Hence $\lceil k/r\rceil=\lceil(k-s)/r\rceil$ if and only
if $s\le b-1$. Since $b\ge1$, we have $(k-1)\bmod r=b-1$, so this is exactly the
condition $s\le(k-1)\bmod r$.

In both cases the equivalence holds, proving the lemma.
\end{proof}

\begin{theorem}\label{cor:two-code-optimality-asym}
Let $\mathcal{Q}_E$ be a pure AEA-qLRC with parameters
$[[n,\kappa,\{d_X,d_Z\},\{r_X,r_Z\};c]]_q$, constructed from
Proposition~\ref{CSS-EA-qLRC-Asym} using two classical
$[n,k_i]_q$ linear codes $\mathcal{C}_i$, $i=1,2$, where
$\mathcal{C}_1$ has distance $d_Z$ and locality $r_Z$, while
$\mathcal{C}_2$ has distance $d_X$ and locality $r_X$. Let $s_1 = \dim(\mathcal{C}_1 \cap \mathcal{C}_2^\perp)$ and $s_2 = \dim(\mathcal{C}_2 \cap \mathcal{C}_1^\perp)$. Then $\mathcal{Q}_E$ is optimal if and only if $\mathcal{C}_1$ and $\mathcal{C}_2$ are optimal LRCs with $d_Z = d_X$ such that$$s_1 \;\le\; (k_1-1) \bmod r_Z \quad \text{and} \quad s_2 \;\le\; (k_2-1) \bmod r_X.$$\end{theorem}
\begin{proof}
By Theorem~\ref{PC-AEAQEC}, the dimension of the AEAQEC code is
$\kappa = k_1 + k_2 - n + c = k_1 - s_1 = k_2 - s_2$. By
Theorem~\ref{pure_EA-qLRC-Asym}, $\mathcal{Q}_E$ is optimal if and only if
$\mathcal{C}_1$ and $\mathcal{C}_2$ are optimal, $d_Z = d_X$, and the ceiling
conditions $\left\lceil\frac{k_1}{r_Z}\right\rceil=\left\lceil\frac{\kappa}{r_Z}\right\rceil$
and $\left\lceil\frac{k_2}{r_X}\right\rceil=\left\lceil\frac{\kappa}{r_X}\right\rceil$
hold. Substituting $\kappa=k_1-s_1$ and $\kappa=k_2-s_2$, these conditions become
\[
\left\lceil\frac{k_1}{r_Z}\right\rceil=\left\lceil\frac{k_1-s_1}{r_Z}\right\rceil
\quad\text{and}\quad
\left\lceil\frac{k_2}{r_X}\right\rceil=\left\lceil\frac{k_2-s_2}{r_X}\right\rceil.
\]
By Lemma~\ref{lem:ceiling} applied with $(k,s,r)=(k_1,s_1,r_Z)$ and
$(k,s,r)=(k_2,s_2,r_X)$ respectively, these two conditions hold if and only if
\[
s_1\le(k_1-1)\bmod r_Z
\quad\text{and}\quad
s_2\le(k_2-1)\bmod r_X.
\]
Therefore, both equality conditions are simultaneously satisfied if and only if
$s_1\le(k_1-1)\bmod r_Z$ and $s_2\le(k_2-1)\bmod r_X$. Hence, the theorem is proved.
\end{proof}

\begin{theorem}\label{EA-qLRC-LCP-Asym}
Let $\mathcal{Q}_E$ be a pure AEA-qLRC with parameters
$[[n,\kappa,\{d_X,d_Z\},\{r_X,r_Z\};c]]_q$, constructed from
Proposition~\ref{CSS-EA-qLRC-Asym} using two classical
$[n,k_i]_q$ linear codes $\mathcal{C}_i$, $i=1,2$, where
$\mathcal{C}_1$ has distance $d_Z$ and locality $r_Z$, while
$\mathcal{C}_2$ has distance $d_X$ and locality $r_X$. Then $\mathcal{Q}_E$ is an optimal maximally entangled pure AEA-qLRC if and only if $\mathcal{C}_1$ and $\mathcal{C}_2$ are optimal LRCs that form a Linear Complementary Pair (LCP), i.e., $k_1=k_2=k$, $d_Z=d_X=d$, and $\mathcal{C}_1 \cap \mathcal{C}_2^\perp = \{0\}$.\end{theorem}

\begin{proof}
By Theorem~\ref{PC-AEAQEC}, an AEAQEC code $\mathcal{Q}_E$ , the number of required EPR pairs is given by $c = \operatorname{rank}(H_1 H_2^T)$. The quantum code $\mathcal{Q}_E$ is maximally entangled if and only if $c$ achieves maximum, which means $c = n-k_1 = n-k_2$. This inherently forces the component code dimensions to be equal, so $k_1 = k_2 = k$.The maximal rank condition $c = n-k$ is achieved if and only if the matrix $H_1 H_2^T$ has full row rank. By the rank-nullity theorem, this occurs exactly when the kernel of the mapping is trivial. As established previously, the kernels correspond to the intersections $\mathcal{C}_1 \cap \mathcal{C}_2^\perp$ and $\mathcal{C}_2 \cap \mathcal{C}_1^\perp$. A trivial kernel means $s_1 = \dim(\mathcal{C}_1 \cap \mathcal{C}_2^\perp) = 0$ and $s_2 = \dim(\mathcal{C}_2 \cap \mathcal{C}_1^\perp) = 0$. This mutual trivial intersection defines $\mathcal{C}_1$ and $\mathcal{C}_2$ as a Linear Complementary Pair (LCP) of codes.Because $\mathcal{C}_1$ and $\mathcal{C}_2$ possess the distinct asymmetric locality parameters $r_Z$ and $r_X$ respectively, we apply Theorem \ref{cor:two-code-optimality-asym},  with $s_1 = 0$ and $s_2 = 0$, the optimality constraints:
 $$s_1 \le (k_1-1) \bmod r_Z \quad \text{and} \quad s_2 \le (k_2-1) \bmod r_X$$are trivially satisfied for $k \ge 1$. Therefore, by Theorem \ref{cor:two-code-optimality-asym}, $\mathcal{Q}_E$ is an optimal maximally entangled EA-qLRC if and only if the constituent codes $\mathcal{C}_1$ and $\mathcal{C}_2$ are optimal classical LRCs that form an LCP with equal minimum distances $d_Z = d_X = d$.\end{proof}

\section{Explicit Constructions with Asymmetric Locality}\label{sec:construction}
In this section, we present explicit AEA-qLRC constructions based on
\emph{Tamo-Barg codes}, which are Reed-Solomon-like polynomial
evaluation codes and \emph{cyclic LRCs}, whose algebraic structure enables simultaneous
control of the minimum distance, locality, and dimension of the
relevant dual-code intersection.

\subsection{Construction from Tamo-Barg Codes}
Let $\mathbb{F}_q$ be a finite field of characteristic $p$, and let $r_Z \ge 2$ and $r_X \ge 2$ be integers such that both $r_Z+1$ and $r_X+1$ divide $q-1$. Let $\ell_1, \ell_2$ be integers satisfying $1 \le \ell_1, \ell_2 < q-1$. We utilize the common evaluation set $\mathcal{P} = \mathbb{F}_q^{*} = \{1, \alpha, \alpha^2, \ldots, \alpha^{q-2}\}$, giving code length $n = \vert{}\mathcal{P}\vert{} = q-1$.To achieve asymmetric locality $\{r_X, r_Z\}$, we construct two distinct Tamo--Barg locally recoverable codes, $\mathcal{C}_1$ and $\mathcal{C}_2$, defined by their respective exponent sets:
  $$\begin{aligned} S_1 &= \{i \in [\ell_1-1]^\dagger : i \not\equiv r_Z \pmod{r_Z+1}\}, \\ S_2 &= \{i \in [\ell_2-1]^\dagger : i \not\equiv r_X \pmod{r_X+1}\}. \end{aligned}$$The dimensions are $k_1 = \ell_1 - \lfloor\ell_1/(r_Z+1)\rfloor$ and $k_2 = \ell_2 - \lfloor\ell_2/(r_X+1)\rfloor$. By the evaluation of the Euclidean inner product derived for the symmetric case, the dual code exponent sets are uniquely determined by $S_i^\perp = \{j \in [n-1] : n-j \notin S_i\}$. Specifically:$$\begin{aligned} S_1^\perp &= \{i \in [q-2] : i \le q-\ell_1-1 \text{ or } i \equiv 1 \pmod{r_Z+1}\}, \\ S_2^\perp &= \{i \in [q-2] : i \le q-\ell_2-1 \text{ or } i \equiv 1 \pmod{r_X+1}\}. \end{aligned}$$

  \begin{proposition}\label{prop:TBintersection-asym}
  Let $q$ be a prime power, $n = q-1$, and $r_Z, r_X \ge 2$ be integers such that both $r_Z+1$ and $r_X+1$ divide $n$. Let $1 \le \ell_1, \ell_2 < n$, and let $S_1 = \{i \in [\ell_1-1]^\dagger : i \not\equiv r_Z \pmod{r_Z+1}\}$ and $S_2^\perp = \{i \in [n-1] : i \le q-\ell_2-1 \text{ or } i \equiv 1 \pmod{r_X+1}\}$ be the respective exponent sets of the code $\mathcal{C}_1$ and the dual of code $\mathcal{C}_2$. Then, $S_1 \cap S_2^\perp$ is given by:$$\begin{cases}  \{i \in [\ell_1-1]^\dagger : i \not\equiv r_Z \pmod{r_Z+1}\}, & \ell_1 + \ell_2 \le q, \\  \{i \in [q-\ell_2-1]^\dagger : i \not\equiv r_Z \pmod{r_Z+1}\} \cup \Lambda_{1,2}, & \ell_1 + \ell_2 > q,  \end{cases}$$where $\Lambda_{1,2} = \{i \in [\ell_1-1]^\dagger : i \ge q-\ell_2,\; i \not\equiv r_Z \pmod{r_Z+1},\; i \equiv 1 \pmod{r_X+1}\}$.Moreover, $s_1 := \vert{}S_1 \cap S_2^\perp\vert{}$ is given by$$s_1 =  \begin{cases} \ell_1 - 1 - \left\lfloor\dfrac{\ell_1}{r_Z+1}\right\rfloor, & \ell_1 + \ell_2 \le q, \\ (q - \ell_2 - 1) - \left\lfloor\dfrac{q-\ell_2}{r_Z+1}\right\rfloor + \vert{}\Lambda_{1,2}\vert{}, & \ell_1 + \ell_2 > q. \end{cases}$$By symmetry, $s_2 := \vert{}S_2 \cap S_1^\perp\vert{}$ is obtained by interchanging the indices $1$ and $2$, and the variables $r_Z$ and $r_X$.
  \end{proposition}
  \begin{proof}
  We decompose $S_2^\perp = S^\perp_{2,A} \cup S^\perp_{2,B}$, where $S^\perp_{2,A} = \{i \in [n-1] : i \le q-\ell_2-1\}$ and $S^\perp_{2,B} = \{i \in [n-1] : i \ge q-\ell_2, \; i \equiv 1 \pmod{r_X+1}\}$. Thus, $S_1 \cap S_2^\perp = (S_1 \cap S^\perp_{2,A}) \cup (S_1 \cap S^\perp_{2,B})$.Observe that $S_1 \cap S^\perp_{2,A} = \{i \in [m]^\dagger : i \not\equiv r_Z \pmod{r_Z+1}\}$, where $m = \min\{\ell_1-1, q-\ell_2-1\}$. Among the integers up to $m$, exactly $\left\lfloor\frac{m+1}{r_Z+1}\right\rfloor$ are congruent to $r_Z \pmod{r_Z+1}$. Therefore, $\vert{}S_1 \cap S^\perp_{2,A}\vert{} = m - \left\lfloor\frac{m+1}{r_Z+1}\right\rfloor$.Next, $S_1 \cap S^\perp_{2,B} = \{i \in [\ell_1-1]^\dagger : i \ge q-\ell_2,\; i \not\equiv r_Z \pmod{r_Z+1},\; i \equiv 1 \pmod{r_X+1}\}$, which is defined as $\Lambda_{1,2}$. Due to the independence of the asymmetric recovery constraints $r_Z$ and $r_X$, this cardinality relies on the greatest common divisor of the cyclic shifts and is denoted as $\vert{}\Lambda_{1,2}\vert{}$.  We consider two cases:\begin{enumerate}\item If $\ell_1 + \ell_2 \le q$, then $\ell_1-1 \le q-\ell_2-1$, so $m = \ell_1-1$ and $S_1 \cap S^\perp_{2,B} = \varnothing$. Hence $S_1 \cap S_2^\perp = S_1 \cap S^\perp_{2,A}$, yielding $s_1 = (\ell_1-1) - \left\lfloor\frac{\ell_1}{r_Z+1}\right\rfloor$.\item If $\ell_1 + \ell_2 > q$, then $q-\ell_2-1 < \ell_1-1$, so $m = q-\ell_2-1$. Hence $S_1 \cap S^\perp_{2,A} = \{i \in [q-\ell_2-1]^\dagger : i \not\equiv r_Z \pmod{r_Z+1}\}$ and $S_1 \cap S^\perp_{2,B} = \Lambda_{1,2}$. Adding the cardinalities of these disjoint sets yields the claimed expression for $s_1$.\end{enumerate}By swapping the role of the $X$-type and $Z$-type evaluation polynomials, an identical argument establishes the formula for $s_2$.
  \end{proof}

  \begin{theorem}\label{thm:TB-EAqLRC-Asym-1}
  Let $q$ be a prime power, and let $r_Z, r_X \ge 2$ be integers such that $(r_Z+1) \mid (q-1)$ and $(r_X+1) \mid (q-1)$. If $1 \le \ell_1, \ell_2 < q-1$ satisfy $\ell_1 + \ell_2 \le q$, then there exists a pure AEA-qLRC with parameters $[[q-1,1,\ge q-\max(\ell_1,\ell_2),\{r_X,r_Z\};(q-1)-k_2-s_1]]_q$, where $k_1 = \ell_1 - \left\lfloor\frac{\ell_1}{r_Z+1}\right\rfloor,\quad k_2 = \ell_2 - \left\lfloor\frac{\ell_2}{r_X+1}\right\rfloor$, and $s_1 = \ell_1 - 1 - \lfloor\ell_1/(r_Z+1)\rfloor$.\end{theorem}
  \begin{proof}
  Follows from Proposition~\ref{CSS-EA-qLRC-Asym} and Proposition~\ref{prop:TBintersection-asym}. The logical dimension is given by $\kappa = k_1 + k_2 - n + c$. Because $c = (q-1) - k_2 - s_1$, the dimension simplifies to $\kappa = k_1 - s_1$. Substituting the expression for $k_1$ and $s_1$ yields $\kappa = 1$. The overall minimum distance is $\delta = \min(d_Z, d_X) \ge \min(q - \ell_1, q - \ell_2) = q - \max(\ell_1, \ell_2)$.
  \end{proof}
  \begin{remark}\label{rem:TB-degeneracy}
In the disjoint-support regime $\ell_1+\ell_2\le q$ of
Theorem~\ref{thm:TB-EAqLRC-Asym-1}, one computes directly from
Proposition~\ref{prop:TBintersection-asym} that
\[
s_1=\ell_1-1-\left\lfloor\frac{\ell_1}{r_Z+1}\right\rfloor=k_1-1
\]
for \emph{every} admissible choice of $\ell_1,\ell_2,r_Z,r_X$ satisfying the
hypotheses of the theorem. Consequently $\kappa=k_1-s_1=1$ identically throughout
this regime: Theorem~\ref{thm:TB-EAqLRC-Asym-1} yields a family of pure AEA-qLRCs
encoding only a single logical qudit, regardless of the length $n=q-1$ or the
localities $r_X,r_Z$ chosen. Constructions with $\kappa>1$ therefore require the
overlapping-support regime $\ell_1+\ell_2>q$ treated in
Theorem~\ref{thm:TB-EAqLRC-Asym-2} below, where $s_1$ depends on the auxiliary set
$\Lambda_{1,2}$ and need not force $\kappa=1$.
\end{remark}
As noted in Remark~\ref{rem:TB-degeneracy}, we now turn to the overlapping-support
regime $\ell_1+\ell_2>q$, which permits logical dimension $\kappa>1$ and therefore
yields a genuinely non-degenerate family of pure AEA-qLRCs.
  \begin{theorem}\label{thm:TB-EAqLRC-Asym-2}Let $q$ be a prime power, and let $r_Z, r_X \ge 2$ be integers such that $(r_Z+1) \mid (q-1)$ and $(r_X+1) \mid (q-1)$. If $1 \le \ell_1, \ell_2 < q-1$ satisfy $\ell_1 + \ell_2 > q$, then there exists a pure AEA-qLRC with parameters $[[q-1,k_1-s_1,\ge\min(q-\ell_1,q-\ell_2),\{r_X,r_Z\};(q-1)-k_2-s_1]]_q$,where $k_1 = \ell_1 - \left\lfloor\frac{\ell_1}{r_Z+1}\right\rfloor,\quad k_2 = \ell_2 - \left\lfloor\frac{\ell_2}{r_X+1}\right\rfloor$ for $i \in \{1,2\}$, and
  $$\begin{aligned} s_1 &= (q - 1 - \ell_2) - \left\lfloor\frac{q-\ell_2}{r_Z+1}\right\rfloor + \vert{}\Lambda_{1,2}\vert{}, \\ \Lambda_{1,2} &= \{i \in [\ell_1-1]^\dagger : i \ge q-\ell_2,\; i \not\equiv r_Z \pmod{r_Z+1},\; i \equiv 1 \pmod{r_X+1}\}. \end{aligned}$$
  \end{theorem}
\begin{proof}
Follows from Proposition~\ref{CSS-EA-qLRC-Asym} and Proposition~\ref{prop:TBintersection-asym}. The intersection cardinality $s_1$ determines both the entanglement requirement $c = n - k_2 - s_1$ and the logical dimension $\kappa = k_1 - s_1$. The minimum distance satisfies $\delta = \min(d_Z, d_X) \ge \min(q - \ell_1, q - \ell_2)$. Because the components $\mathcal{C}_1$ and $\mathcal{C}_2$ encode the independent recovery parameters $r_Z$ and $r_X$, the resulting code inherits the target asymmetric locality $\{r_X, r_Z\}$.
\end{proof}

\subsection{Construction from Cyclic Codes}

\begin{lemma}\cite[Lemma 8]{kumar2026entanglement}\label{Hull dim}
Let $\mathcal{C}_1$ and $\mathcal{C}_2$ be cyclic codes of length $n$ over $\mathbb{F}_q$ (respectively, over $\mathbb{F}_{q^2}$ for the Hermitian case) with defining sets $Z_1, Z_2 \subseteq \mathbb{Z}_n$, and dimensions $k_i = n - |Z_i|$ for $i = 1, 2$.
\begin{enumerate}
    \item[(1)] $\dim\left(\mathcal{C}_1^{\perp} \cap \mathcal{C}_2\right) = |Z_1| - |Z_2 \cap (-Z_1)|$, where $-Z_1 := \{-i \bmod n : i \in Z_1\}$.
    \item[(2)] If $\mathcal{C}_1, \mathcal{C}_2$ are cyclic over $\mathbb{F}_{q^2}$, then $\dim\left(\mathcal{C}_1^{\perp_h} \cap \mathcal{C}_2\right) = |Z_1| - |Z_2 \cap (-qZ_1)|$, where $-qZ_1 := \{-qi \bmod n : i \in Z_1\}$.
\end{enumerate}
\end{lemma}

\begin{theorem}\cite[Theorem 11]{kumar2026entanglement}\label{Cyclic EAQEC}
Let $\mathcal{C}_1$ and $\mathcal{C}_2$ be cyclic codes over
$\mathbb{F}_q$ with parameters $[n,k_1,d_Z]_q$ and
$[n,k_2,d_X]_q$, respectively, and defining sets
$Z_1,Z_2\subseteq\mathbb{Z}_n$. Then there exists a pure EAQEC code with parameters
$[[n, \; n - |(-Z_1) \cup Z_2|, \; \{d_X, d_Z\}; \; |(-Z_1) \cap Z_2|]]_q$.
\end{theorem}

\begin{lemma}\label{lem:support-alignment-asym}Let $n$ be such that $(r_Z+1)\mid n$ and $(r_X+1)\mid n$. Set $\nu_Z=n/(r_Z+1)$ and $\nu_X=n/(r_X+1)$, and let $\alpha$ be a primitive $n$-th root of unity. For $\ell_1\in[r_Z]^{\dagger}$ and $\ell_2\in[r_X]^{\dagger}$, let $\mathcal{C}_1$ and $\mathcal{C}_2$ be cyclic codes of length $n$ over $\mathbb{F}_q$ whose defining sets contain:$$L_1=\{\,j\in[n-1]^{\dagger}: j\equiv\ell_1 \!\!\pmod{r_Z+1}\,\}$$$$L_2=\{\,j\in[n-1]^{\dagger}: j\equiv\ell_2 \!\!\pmod{r_X+1}\,\}$$respectively. Then for each $u_1\in[\nu_Z-1]^{\dagger}$ and $u_2\in[\nu_X-1]^{\dagger}$, the vectors $h_{u_1}^{(\ell_1)}\in\mathbb{F}_q^{\,n}$ and $h_{u_2}^{(\ell_2)}\in\mathbb{F}_q^{\,n}$ with entries$$\bigl(h_{u_1}^{(\ell_1)}\bigr)_m= \begin{cases} \alpha^{\ell_1 m}, & m\equiv u_1 \!\!\pmod{\nu_Z},\\ 0, & \text{otherwise,} \end{cases}$$and$$\bigl(h_{u_2}^{(\ell_2)}\bigr)_m= \begin{cases} \alpha^{\ell_2 m}, & m\equiv u_2 \!\!\pmod{\nu_X},\\ 0, & \text{otherwise,} \end{cases}$$lie in $\mathcal{C}_1^{\perp}$ and $\mathcal{C}_2^{\perp}$, respectively. They satisfy:$$\mathrm{supp}\bigl(h_{u_1}^{(\ell_1)}\bigr)=\{\,m\in[n-1]^{\dagger}: m\equiv u_1\!\!\pmod{\nu_Z}\,\},$$$$\mathrm{supp}\bigl(h_{u_2}^{(\ell_2)}\bigr)=\{\,m\in[n-1]^{\dagger}: m\equiv u_2\!\!\pmod{\nu_X}\,\},$$which are sets of sizes $r_Z+1$ and $r_X+1$, respectively.Consequently, for every coordinate $i\in[n]$, the checks $h_{u_1(i)}^{(\ell_1)}\in\mathcal{C}_1^{\perp}$ and $h_{u_2(i)}^{(\ell_2)}\in\mathcal{C}_2^{\perp}$, where $u_1(i)\equiv i \pmod{\nu_Z}$ and $u_2(i)\equiv i \pmod{\nu_X}$, satisfy$$i\in\mathrm{supp}\bigl(h_{u_1(i)}^{(\ell_1)}\bigr)\cap\mathrm{supp}\bigl(h_{u_2(i)}^{(\ell_2)}\bigr),$$with $\bigl\vert{}\mathrm{supp}\bigl(h_{u_1(i)}^{(\ell_1)}\bigr)\bigr\vert{}=r_Z+1$ and $\bigl\vert{}\mathrm{supp}\bigl(h_{u_2(i)}^{(\ell_2)}\bigr)\bigr\vert{}=r_X+1$.\end{lemma}\begin{proof}We demonstrate the property for $\mathcal{C}_1$; the proof for $\mathcal{C}_2$ follows identically by replacing the index $1$ with $2$ and $Z$ with $X$. Since $L_1$ is contained in the defining set, every $c=(c_0,\dots,c_{n-1})\in\mathcal{C}_1$ satisfies $c(\alpha^{\ell_1+t(r_Z+1)})=0$ for $t\in[\nu_Z-1]^{\dagger}$. Put $\beta_Z=\alpha^{r_Z+1}$, an element of order $\nu_Z$, and $S_{u_1}=\sum_{m\equiv u_1\,(\mathrm{mod}\,\nu_Z)}c_m\alpha^{\ell_1 m}$ for $u_1\in[\nu_Z-1]^{\dagger}$.Grouping the coordinates by residue modulo $\nu_Z$ gives$$0=c\bigl(\alpha^{\ell_1+t(r_Z+1)}\bigr)=\sum_{m=0}^{n-1}c_m\alpha^{\ell_1 m}\beta_Z^{tm} =\sum_{u_1=0}^{\nu_Z-1}\beta_Z^{tu_1}S_{u_1} ,\qquad t\in[\nu_Z-1]^{\dagger},$$using $\beta_Z^{tm}=\beta_Z^{tu_1}$ whenever $m\equiv u_1\pmod{\nu_Z}$. The coefficient matrix $(\beta_Z^{tu_1})_{0\le t,u_1\le\nu_Z-1}$ is Vandermonde in the distinct powers $1,\beta_Z,\dots,\beta_Z^{\nu_Z-1}$, hence invertible, so $S_{u_1}=0$ for every $u_1$. As $S_{u_1}=\langle c,h_{u_1}^{(\ell_1)}\rangle$, this establishes $h_{u_1}^{(\ell_1)}\in\mathcal{C}_1^{\perp}$.Each residue class modulo $\nu_Z$ meets $[n-1]^{\dagger}$ in exactly $n/\nu_Z=r_Z+1$ coordinates, and the entries $\alpha^{\ell_1 m}$ are nonzero, so $\mathrm{supp}(h_{u_1}^{(\ell_1)})$ is exactly that residue class, independent of $\ell_1$. The final display follows by taking $u_1(i)\equiv i\pmod{\nu_Z}$ and $u_2(i)\equiv i\pmod{\nu_X}$ in their respective codes, placing $i$ simultaneously in both supports.\end{proof}

\begin{theorem}\label{Cyclic-EA-qLRC-Asym}
Let $(r_Z+1)\mid n$ and $(r_X+1)\mid n$. Let $\mathcal{C}_1$ and $\mathcal{C}_2$ be cyclic LRC codes over $\mathbb{F}_q$ with localities $r_Z$ and $r_X$ respectively, and defining sets $Z_i=L_i\cup D_i$ for $i\in\{1,2\}$, constructed as in definition~\ref{cyclicLRC} with respect to a common primitive $n$-th root of unity $\alpha$. Then there exists a pure AEA-qLRC with parameters $[[n, \; n - |(-Z_1) \cup Z_2|, \; \{d_X, d_Z\},\{r_X, r_Z\}; \; |(-Z_1) \cap Z_2|]]_q$.
\end{theorem}
\begin{proof}
By construction, $L_i\subseteq Z_i$ for $i=1,2$. Therefore, the asymmetric extension of Lemma~\ref{lem:support-alignment-asym} applies to each $\mathcal{C}_i$ and supplies, for every coordinate $i\in[n]$, parity checks $c_1^{(i)}\in\mathcal{C}_1^{\perp}$ and $c_2^{(i)}\in\mathcal{C}_2^{\perp}$ such that $i\in\mathrm{supp}(c_1^{(i)})\cap\mathrm{supp}(c_2^{(i)})$. The independent support sizes satisfy $\vert{}\mathrm{supp}(c_1^{(i)})\vert{}=r_Z+1$ and $\vert{}\mathrm{supp}(c_2^{(i)})\vert{}=r_X+1$. The asymmetric CSS-like criterion (Proposition~\ref{CSS-EA-qLRC-Asym}) is therefore satisfied using the independent recovery sets $\Gamma_i^{(Z)}=\{\,m: m\equiv i \!\!\pmod{n/(r_Z+1)}\,\}$ and $\Gamma_i^{(X)}=\{\,m: m\equiv i \!\!\pmod{n/(r_X+1)}\,\}$. Because $\mathcal{C}_1$ governs the Z-type recovery and $\mathcal{C}_2$ governs the X-type recovery, the resulting quantum code possesses asymmetric locality $\{r_X,r_Z\}$. The stated dimension, distance, and entanglement parameters follow directly from Theorem~\ref{Cyclic EAQEC}.
\end{proof}

\begin{example}
 To illustrate the construction over the base field $q=16$ with length $n=15$, where the $16$-cyclotomic cosets modulo $15$ all have a size of $1$, we target asymmetric localities of $r_X=2$ and $r_Z=4$ since $(r_X+1)=3$ and $(r_Z+1)=5$ both divide $n=15$. Using a primitive $15$-th root of unity $\alpha \in \mathbb{F}_{16}$, we construct two constituent cyclic LRCs, $\mathcal{C}_1$ and $\mathcal{C}_2$. For the Z-recovery code $\mathcal{C}_1$ with $r_Z=4$, setting $\ell_1=1$ establishes the locality subset $L_1=\{1, 6, 11\}$, and appending $d_Z=\{2, 3\}$ yields the defining set $Z_1=\{1, 2, 3, 6, 11\}$, which results in a dimension $k_1=15-\vert{}Z_1\vert{}=10$ and a BCH guaranteed distance $d_Z \ge 4$. For the X-recovery code $\mathcal{C}_2$ with $r_X=2$, choosing $\ell_2=1$ provides the locality subset $L_2=\{1, 4, 7, 10, 13\}$, and appending $d_X=\{2\}$ gives the defining set $Z_2=\{1, 2, 4, 7, 10, 13\}$, resulting in a dimension $k_2=15-\vert{}Z_2\vert{}=9$ and a BCH distance $d_X \ge 3$. To determine the quantum code parameters, we compute the inverse set $-Z_1 \equiv \{14, 13, 12, 9, 4\} \pmod{15}$, which dictates the required entanglement pairs $c = \vert{}(-Z_1) \cap Z_2\vert{} = \vert{}\{4, 13\}\vert{} = 2$ and the logical dimension $\kappa = 15 - \vert{}(-Z_1) \cup Z_2\vert{} = 15 - 9 = 6$. With a quantum distance $\{d_X, d_Z\} = \{4,3\}$, this yields a pure $[[15, 6, \{4,3\},\{2, 4\}; 2]]_{16}$ AEA-qLRC.
\end{example}
\begin{theorem}\label{Cyclic-MEAQECC-Asym}
Let $\mathcal{C}_1$ and $\mathcal{C}_2$ be an optimal Linear Complementary Pair (LCP) of cyclic LRCs over $\mathbb{F}_q$ with parameters $[n,k,d]_q$, defining sets $Z_1$ and $Z_2$, and localities $r_Z$ and $r_X$, respectively. Then there exists an optimal maximally entangled pure AEA-qLRC with asymmetric locality $\{r_X, r_Z\}$ and parameters $[[n,k,d,\{r_X,r_Z\};n-k]]_q$.
\end{theorem}
\begin{proof}
Since $\mathcal{C}_1$ and $\mathcal{C}_2$ form an LCP, we have
$\mathcal{C}_1\cap\mathcal{C}_2^\perp=\{0\}; k_1=k_2=k,\;d_Z=d_X=d.$
By Theorem~\ref{Cyclic-EA-qLRC-Asym}, these cyclic LRCs induce a pure  AEA-qLRC. Furthermore, by Theorem~\ref{EA-qLRC-LCP-Asym}, the LCP condition gives $c=n-k.$
Since $\mathcal{C}_1$ and $\mathcal{C}_2$ are optimal, the resulting quantum code is optimal. Hence there exists a maximally entangled pure AEA-qLRC with parameters $[[n,k,d,\{r_X,r_Z\};n-k]]_q$. 
\end{proof}

\section{Conclusion}\label{sec:conclusion}

In this paper, we introduced entanglement-assisted quantum locally recoverable codes with asymmetric locality $\{r_X,r_Z\}$ , a framework unifying two distinct lines of research in quantum error correction: the asymmetry between bit-flip and phase-flip protection and the local recoverability of individual qudits. We established that this property admits both a stabilizer-based sufficient criterion and, within the CSS framework, a corresponding sufficient condition in terms of independent support-size conditions on the dual codes of the two constituent classical codes. We emphasize that both results are stated as sufficient conditions; establishing matching necessary conditions, and hence a full characterization, is left for future work (see the discussion below). On the basis of this characterization, we derived Singleton-, Griesmer-, Plotkin-, sphere-packing-, and Cadambe--Mazumdar-like bounds for CSS-type AEA-qLRCs (Proposition~\ref{CSS-EA-qLRC-Asym}), and obtained a complete characterization of when a pure CSS-type AEA-qLRC constructed from two classical LRCs attains the Singleton-like bound~\eqref{eqn:Singleton Bound
AEA-qLRC} with equality, expressed in terms of the optimality of the constituent codes and the dimensions of their mutual dual intersections. We further exhibited explicit families of AEA-qLRCs arising from Tamo--Barg codes and from cyclic codes, including a family of maximally entangled codes obtained from linear complementary pairs. Since the symmetric-locality case is recovered as the special instance $r_X=r_Z$, these results strictly extend the known parameter space of entanglement-assisted quantum locally recoverable codes.

Several directions remain open for future investigation. First, the bounds and optimality conditions established here are specific to the CSS-type construction; extending this analysis to general, non-CSS entanglement-assisted stabilizer codes with asymmetric locality would enlarge the attainable parameter space beyond what is captured by the present framework. Second, it would be of interest to incorporate \emph{availability}—that is, the existence of multiple disjoint recovery sets per coordinate—into the asymmetric locality parameter $\{r_X, r_Z\}$, yielding a still more general class of codes. Third, the problem of constructing AEA-qLRCs that attain the Griesmer-, Plotkin-, and sphere-packing-like bounds established in this paper remains open; in particular, determining whether infinite families of AEA-qLRCs achieving these bounds exist constitutes a natural and promising direction for future research.
\bibliography{sn-bibliography}

\end{document}